\documentclass[a4paper,USenglish,dvipsnames]{article}
\usepackage[breaklinks,hidelinks]{hyperref} 
\usepackage{fullpage}
\usepackage{amssymb}
\usepackage{amsmath}
\usepackage{amsthm}
\usepackage{thmtools}
\usepackage{thm-restate}
\usepackage{mathtools}

\newtheorem{remark}{Remark}
\newtheorem{lemma}{Lemma}
\newtheorem{theorem}{Theorem}
\newtheorem{proposition}{Proposition}
\newtheorem{example}{Example}
\newtheorem{definition}{Definition}

\usepackage[utf8]{inputenc}

\usepackage[capitalise,nameinlink]{cleveref}
\crefname{equation}{Eq.}{Eqs.}
\crefname{claim}{Claim}{Claims}

\usepackage{anyfontsize}

\usepackage[xcolor,notion,quotation,electronic]{knowledge}
\usepackage{xparse}
\usepackage{xspace}

\usepackage{bbm}

\usepackage{turnstile}

\usepackage{booktabs}
\usepackage{multirow}
\usepackage{pifont}

\usepackage{multicol}
\usepackage{stackengine}
\usepackage{pgf}

\usepackage{tikz}
\usetikzlibrary{decorations.pathreplacing,positioning}
\usetikzlibrary{arrows,automata,calc}
\tikzset{>=stealth, shorten >=1pt}
\tikzset{every edge/.style = {thick, ->, draw}}
\tikzset{every loop/.style = {thick, ->, draw}}
\pgfdeclarelayer{background}
\pgfsetlayers{background,main}

\usetikzlibrary{
  calc,
  arrows,
  positioning,
  shapes,
  shadows,
  automata,
  external,
  shadows,
  decorations.pathreplacing,
  decorations.pathmorphing,
  decorations.markings,
  backgrounds,
  patterns
}

\tikzset{  
  node distance=4em,
  >=stealth',
  shadow/.style		= {opacity=.25, black, shadow xshift=0.08, shadow yshift=-0.08},
  plainnode/.style 	= {draw, ultra thick, fill=gray!10},
  pl0/.style	       = {circle, minimum size=6mm, inner sep=0mm,plainnode, drop shadow = {shadow}},
  ind/.style         = {circle,draw,fill=black},
  nind/.style        = {circle,draw},
  wt/.style          = {text = white, yshift = -0.1em,scale=0.9}
}

\newtheorem{assumption}[theorem]{Assumption}

\newcommand{\nats}{\mathbbm{N}}

\renewcommand{\epsilon}{\varepsilon}
\renewcommand{\phi}{\varphi}
\newcommand{\myphi}{\kl[myphi]{\phi}}
\newcommand{\mypsi}{\kl[mypsi]{\psi}}
\newcommand{\size}[1]{|#1|}
\newcommand{\pow}[1]{2^{#1}}
\newcommand{\cceq}{\mathop{::=}}
\newcommand{\set}[1]{\{#1\}}

\newcommand{\F}{\mathop{\mathbf{F}\vphantom{a}}\nolimits}
\newcommand{\G}{\mathop{\mathbf{G}\vphantom{a}}\nolimits}
\DeclareMathOperator{\U}{\mathbf{U}}
\newcommand{\X}{\mathop{\mathbf{X}\vphantom{a}}\nolimits}

\newcommand{\ltl}{{LTL}\xspace}

\newcommand{\ctlstar}{{CTL$^*$}\xspace}
\newcommand{\hyltl}{\kl[HyperLTL]{HyperLTL}\xspace}
\newcommand{\hyrechml}{{HyperRecHML}\xspace}
\newcommand{\hyqptl}{\kl[HyperQPTL]{HyperQPTL}\xspace}
\newcommand{\hyctlstar}{{HyperCTL$^*$}\xspace}

\newcommand{\qptl}{\kl[QPTL]{QPTL}\xspace}

\knowledgenewcommand{\var}{\cmdkl{\mathcal{V}}}
\newcommand{\ap}[0]{\mathrm{AP}}

\newcommand{\existsqp}{\tilde{\exists}}
\newcommand{\forallqp}{\tilde{\forall}}

\newcommand{\prophy}{\mathfrak{P}}
\newcommand\pp{\textup{P}}
\newcommand\calP{\mathcal{P}}

\newcommand{\tower}{\textsc{Tower}\xspace}

\newcommand{\myquot}[1]{``#1''}
\newcommand{\flip}[1]{\overline{#1}}

\newcommand{\tsys}{\mathfrak{T}}
\newcommand{\mytsys}{\kl[mytsys]{\mathfrak{T}}}
\knowledgenewrobustcmd{\traces}[1]{\cmdkl{\mathrm{Tr}}(#1)}
\newcommand{\Succ}[1]{\mathrm{S}(#1)}

\newcommand{\ieven}{\kl[index]{I_{\forall}}}
\newcommand{\iodd}{\kl[index]{I_{\exists}}}

\NewDocumentCommand{\tsysmanip}{O{\calP}}{%
	 \kl[system manipulation]{\tsys^{#1}}}%
\NewDocumentCommand{\phimanip}{O{\calP}O{\Xi}}{%
	 \kl[property manipulation]{\ensuremath{\phi^{#1,#2}}}}%
\newcommand{\psimanip}{%
	 \kl[psimanip]{\ensuremath{\psi^{\calP,\Xi}}}
}%

\knowledgenewrobustcmd{\prophecies}[1]{%
	 \ensuremath{\cmdkl{\mathrm{Prophecies}}(#1)}%
}%
\knowledgenewrobustcmd{\orig}[1]{%
	 \ensuremath{\cmdkl{\lambda_{\ap}}(#1)}%
}%
\knowledgenewrobustcmd{\marking}[1]{%
	 \ensuremath{\cmdkl{\mathrm{mark}}(#1)}%
}%

\newcommand{\aut}{\mathcal{A}}
\newcommand{\autb}{\mathcal{B}}

\newcommand{\initmark}{I}
\newcommand{\col}{\Omega}

\NewDocumentCommand{\auti}{mo}{%
 \IfNoValueTF{#2}%
    {\ensuremath{\kl[aut]{\mathcal{A}^{\tsys}_{#1}}}}%
    {\ensuremath{\kl[aut]{(\mathcal{A}_{#1}^{\tsys})_{#2}}}}%
}%

\knowledgenewrobustcmd{\combine}[1]{%
	 \cmdkl{\mathrm{mrg}}(#1)}%

\newcommand{\vectordots}[2]{
\begin{pmatrix}
  #1 \\
    \vspace{-.5cm}\\
  \vdots \\
    \vspace{-.5cm}\\
  #2\\
\end{pmatrix}}

\renewcommand{\vectordots}[2]{( #1 , \ldots , #2 )}

\knowledgenewrobustcmd{\equivrel}[1]{%
	 \cmdkl{\equiv_{#1}}}%

\newcommand{\tr}[1]{
\def\a{t}\def\b{#1}
\ifx\a\b{\bm\tilde{t}}\else\bm\tilde{#1}\fi
}

\newcommand{\dom}[1]{\mathrm{dom}(#1)}

\NewDocumentCommand{\safeprophy}{mmm}{%
 \ensuremath{\kl[safeprophy]{(\mathfrak{S}_{#1})_{#2}^{#3}}}%
}%

\knowledgenewrobustcmd{\progressprophy}[3]{%
 \ensuremath{\cmdkl{(\mathfrak{P}_{#1})_{#2}^{#3}}}%
}%

\knowledgenewrobustcmd{\psafe}[3]{%
 \ensuremath{\cmdkl{(\texttt{p}_{#1})_{#2}^{#3}}}%
}%

\knowledgenewrobustcmd{\pprogress}[3]{%
 \ensuremath{\cmdkl{(\texttt{p}_{#1})_{#2}^{#3}}}%
}%

\newcommand{\pformula}[3]{%
 \ensuremath{(\xi_{#1})_{#2}^{#3}}%
}%

\newcommand{\vplayer}[1]{\kl[player]{Verifier~$#1$}\xspace}
\knowledgenewrobustcmd{\game}[1]{\cmdkl{\mathcal{G}}(#1)}

\NewDocumentCommand{\round}{mo}{%
 \IfNoValueTF{#2}%
    {\kl[round]{Round~$#1$}\xspace}%
    {\kl[round]{Round~$(#1,#2)$}\xspace}%
}%

\newcommand{\hist}{\mathrm{Hist}}

\newcommand{\propo}{\texttt{p}}
\newcommand{\propovar}{\texttt{q}}
\newcommand{\merge}{^\smallfrown}

\newcommand{\gni}{\mathrm{GNI}}
\newcommand{\inp}{\mathrm{in}}
\newcommand{\out}{\mathrm{out}}
\newcommand{\inc}{\mathrm{inc}}
\newcommand{\all}{\mathrm{all}}

\newcommand{\markprop}{\texttt{m}}

\newrobustcmd{\possmoves}[2]{%
    \ensuremath{M_{#1}}(#2)}%

\knowledgenewrobustcmd{\progressmoves}[1]{%
    \ensuremath{\cmdkl{\mathrm{Prog}}(#1)}}%

\newcommand{\emptycase}{\textnormal{\textsc{Empty}}}
\newcommand{\nonemptycase}{\textnormal{\textsc{Nonempty}}}

\knowledgenewrobustcmd{\lastb}[1]{\cmdkl{\mathrm{last}_B}(#1)}
\knowledgenewrobustcmd{\firstf}[1]{\cmdkl{\mathrm{first}_F}(#1)}
\knowledgenewrobustcmd{\val}[1]{\cmdkl{\mathrm{val}}(#1)}

\knowledgenewrobustcmd{\opt}[2]{\cmdkl{\mathrm{opt}_{#1}}(#2)}

\knowledge{notion}
	| closed

\knowledge{notion}
	| merge
    
\knowledge{notion}
	| transition system

\knowledge{notion}
	| path
    | paths

\knowledge{notion}
	| trace
    | traces

\knowledge{notion}
	| set of traces

\knowledge{notion}
	| parity automaton
    | parity automata
    | parity 

\knowledge{notion}
	| safety automaton
    | safety automata
    | safety

\knowledge{notion}
	| HyperLTL

\knowledge{notion}
	| QPTL

\knowledge{notion}
	| HyperQPTL

\knowledge{notion}
    | trace variable
    | trace variables
    | variable
    | variables

\knowledge{notion}
	| trace assignment
    | trace assignments
    | (trace) variable assignment
    | variable assignment
    | variable assignments
    
\knowledge{notion}
	| satisfy
    | satisfies

\knowledge{notion}
	| model
    | models

\knowledge{notion}
	| proposition assignment
    | proposition assignments

\knowledge{ignore}
	| Skolem function
    | Skolem functions

\knowledge{notion}
	| mygame

\knowledge{notion}
	| mytsys

\knowledge{notion}
	| myphi

\knowledge{notion}
	| mypsi

\knowledge{notion}
	| psimanip

\knowledge{notion}
	| aut

\knowledge{notion}
	| safeprophy

\knowledge{notion}
	| prophyall

\knowledge{notion}
	| prophecy
    | Prophecy
    | prophecies
    | Prophecies

\knowledge{notion}
	| prophecy variable
    | prophecy variables

\knowledge{notion}
	| System manipulation
	| system manipulation

\knowledge{notion}
	| Property manipulation
	| property manipulation

\knowledge{notion}
	| Round
	| round

\knowledge{notion}
	| Player
	| player
    | players
	| Verifier-player
	| Verifier-players
    | Falsifier
    | Verifier

\knowledge{notion}
	| index

\knowledge{notion}
	| Rabin
    | Rabin automata
    | Rabin automaton

\knowledgestyle*{intro notion}{color=nicecyan, emphasize}
\knowledgestyle*{notion}{color=darkblue}
\knowledgestyle*{unknown}{color=nicered}

\IfKnowledgePaperModeTF{
\knowledgestyle*{intro notion}{color=black, emphasize}
\knowledgestyle*{notion}{color=black}
\knowledgestyle*{unknown}{color=black}
}{}

\definecolor{darkblue}{RGB}{0,0,92}
\definecolor{nicecyan}{HTML}{006165}
\definecolor{nicered}{HTML}{DB3A34}
\definecolor{nicegreen}{HTML}{6D972E}

\usepackage{amsthm}

\usepackage[color=Cyan!70]{todonotes}
\usepackage{bm}
\usepackage{xfakebold}

\title{Prophecies all the Way:\\ Game-based Model-Checking for HyperQPTL beyond $\forall^*\exists^*$}
\author{Sarah Winter (IRIF, Université Paris Cité, Paris, France)\\ Martin Zimmermann (Aalborg University, Aalborg, Denmark)}
\date{}

\begin{document}

\maketitle

\begin{abstract}
    Model-checking HyperLTL, a temporal logic expressing properties of sets of traces with applications to information-flow based security and privacy, has a decidable, but TOWER-complete, model-checking problem.
While the classical model-checking algorithm for full HyperLTL is automata-theoretic, more recently, a game-based alternative for the $\forall^*\exists^*$-fragment has been presented.

Here, we employ imperfect information-games to extend the game-based approach to full HyperQPTL, which features arbitrary quantifier prefixes and quantification over propositions and can express every $\omega$-regular hyperproperty.
As a byproduct of our game-based algorithm, we obtain finite-state implementations of Skolem functions via transducers with lookahead that explain satisfaction or violation of HyperQPTL properties.
\end{abstract}

\section{Introduction}

Hyperlogics like \hyltl and \hyctlstar~\cite{ClarksonFKMRS14} extend their classical counterparts \ltl~\cite{Pnueli77} and \ctlstar~\cite{EmersonH86} by trace quantification and are thereby able to express so-called hyperproperties~\cite{ClarksonS10}, properties which relate multiple execution traces of a system. 
These are crucial for expressing information-flow properties capturing security and privacy requirements~\cite{ClarksonFKMRS14}.
For example, generalized noninterference~\cite{gni}  is captured by the \hyltl formula
\[
\phi_\gni = \forall \pi.\ \forall \pi'.\ \exists \pi''.\ \G \left(\bigwedge\nolimits_{\propo \in L_{\inp} \cup L_{\out}} \propo_\pi \leftrightarrow \propo_{\pi''}\right) \wedge \G\left( \bigwedge\nolimits_{\propo \in H_\inp} \propo_{\pi'} \leftrightarrow \propo_{\pi''} \right)
\]
expressing that for all traces $\pi$ and $\pi'$ there exists a trace~$\pi''$ that agrees with the low-security inputs (propositions in $L_\inp$) and low-security outputs (propositions in $L_\out$) of $\pi$ and the high-security inputs (propositions in $H_\inp$) of $\pi'$. 
Intuitively, it is satisfied if every input-output behavior observable by a low-security user of a system
is compatible with any sequence of high-security inputs, i.e., the low security input-output behavior does not leak information about the high-security inputs, which should indeed be unobservable (directly and indirectly) by a low-security user.

These expressive logics offer a uniform approach to the specification, analysis, and verification of hyperproperties  with intuitive syntax and semantics, and a decidable~\cite{ClarksonFKMRS14}, albeit \tower-complete~\cite{Rabe16diss,MZ20}, model-checking problem.
The classical model-checking algorithm~\cite{ClarksonFKMRS14} is automata-based and even the case of $\forall^*\exists^*$-formulas like $\phi_\gni$ requires the complementation of $\omega$-automata, a famously challenging construction to implement.

As an alternative, it is beneficial to draw upon the deep connections between logic and games. 
Consider a \hyltl formula of the form~$\forall \pi.\ \exists \pi'.\psi$ with quantifier-free $\psi$. 
It is straightforward to capture the semantics of \hyltl by a model-checking game between two players called Verifier (trying to prove $\tsys \models \phi$) and Falsifier (trying to prove $\tsys \not\models \phi$).
In the model-checking game, Falsifier picks a trace of $\tsys$ to interpret $\pi$ and then Verifier picks a trace of $\tsys$ to interpret $\pi'$. Verifier wins if these two traces satisfy the formula~$\psi$, otherwise Falsifier wins. 
The game can easily be shown sound (if Verifier wins, then $\tsys \models\phi$) and complete (if $\tsys\models\phi$ then Verifier wins), as winning strategies for Verifier are Skolem functions for $\pi'$ and vice versa.
But the game has one major drawback: both players have in general uncountably many possible moves as they are picking traces.

To obtain a game that can be handled algorithmically, Coenen et al.\ introduced the following variation of the model-checking game~\cite{hyperliveness} (which we will call the alternating game):
Instead of picking the traces of $\tsys$ in one move, the players construct them alternatingly one vertex at a time. 
Together with a deterministic $\omega$-automaton that is equivalent to $\psi$, this yields a finite game with $\omega$-regular winning condition that is sound, but possibly incomplete ($\tsys \models \phi$ does not always imply that Verifier wins the alternating game).

The problem boils down to the informedness of Verifier: In the model-checking game, she has knowledge about the full trace assigned to $\pi$ when picking $\pi'$. 
On the other hand, in the alternating game, Verifier only has access to the first vertex of the trace assigned to $\pi$ when picking the first vertex of the trace assigned to $\pi'$, which puts her at a disadvantage.
That the alternating game can be incomplete is witnessed by the formula
\[
\phi_\inc = \forall \pi.\ \exists \pi'.\ \propo_{\pi'} \leftrightarrow \F \propo_\pi
\]
expressing that for every trace~$\pi$ in $\tsys$ there is a trace~$\pi'$ in $\tsys$ such that $\propo$ holds at the first position of $\pi'$ if and only if there is a $\propo$ somewhere in $\pi$, i.e., the first letter of $\pi'$ depends (possibly) on every letter of $\pi$.
So, by not picking a $\propo$ in the first round, Falsifier forces Verifier to make a prediction about whether Falsifier will ever pick a $\propo$ or not in the future. 
However, Falsifier can easily contradict that prediction and thereby win, e.g., for a transition system~$\tsys_\all$ having all traces over~$\propo$. 
Thus, we indeed have $\tsys_\all \models \phi$, but Verifier does not win the alternating game.

However, a single bit of lookahead allows Verifier to win the alternating game induced by $\tsys_\all$ and $\phi_\inc$, i.e., the answer to the query \myquot{will the trace picked by Falsifier contain a $\propo$}. 
If Falsifier has to provide this information with his first move, then Verifier can make her first move accordingly. 
For correctness, we  additionally have to adapt the rules of the alternating game so that Falsifier loses when contradicting his answer made during the first round, e.g., if he commits to there being no $\propo$'s in his trace, but then picking one. 

In general, it is not sufficient to ask a single query at the beginning of a play, but one needs to ask queries at every move of Falsifier.
To see this, let us consider another example, this time with the formula
\[
\phi_\inc' = \forall \pi.\ \exists \pi'.\  \G (\propo_{\pi'} \leftrightarrow \X \propo_\pi).
\]
Here, Verifier has to always pick a $\propo$ in her move if and only if Falsifier will pick a $\propo$ in his next move (which Verifier does not have yet access to). 
Thus, Verifier wins if she gets the (truthful) answer to the query \myquot{does the next move by Falsifier contain a $\propo$}, but loses without the ability to query Falsifier.
Note that in both cases, the queries are used to obtain a (binding) commitment about the future moves that Falsifier will make.

Coenen et al.~\cite{hyperliveness} and Beutner and Finkbeiner~\cite{BF} showed how to formalize this intuition using so-called prophecies~\cite{DBLP:journals/tcs/AbadiL91}.
A single prophecy is a language~$P$ of traces and is associated with a (Boolean) prophecy variable. 
Now, in addition to picking, vertex by vertex, a trace of $\tsys$, Falsifier also picks a truth value for the prophecy variable with the interpretation that a value of $1$ corresponds to the suffix of the trace picked by him from now on being in $P$ and that a value of of $0$ corresponds to the suffix of the trace picked by him from now on not being in $P$.
If the language~$P$ is $\omega$-regular, then one can employ an $\omega$-automaton to check whether the predictions made by Falsifier are actually truthful (and make him lose if they are not truthful).
We call the resulting game the alternating game with prophecies.
Intuitively, it constitutes a middle-ground between the naive game where both players pick traces and the alternating game (without prophecies). 
The former is sound and complete, but not finite-state while the latter is sound and finite-state, but not complete.

The alternating game with (finitely many) $\omega$-regular prophecies is sound and finite-state and Beutner and Finkbeiner showed that for every $\forall^*\exists^*$-formula there is a finite and effectively computable set of $\omega$-regular prophecies such that Verifier wins the alternating game with these prophecies if and only if $\tsys \models \phi$, i.e., there are always prophecies that make the alternating game-based approach to model-checking also complete.
Furthermore, the resulting game with prophecies has a $\omega$-regular winning condition, and can therefore be solved effectively.
Note that this construction is still automata-based, as the prophecies are derived from a (deterministic) $\omega$-automaton for the quantifier-free part of $\phi$.
Beutner and Finkbeiner implemented their prophecy-based model-checking algorithm and presented encouraging results on small instances, e.g., their prototype implementation is in some cases the first tool that can prove $\tsys\not\models\phi$, a result that was out of reach for existing model-checking tools~\cite{BF}. These results shows the potential of the game-based approach to \hyltl model-checking.

However, one question remains open: can the alternating game-based approach be extended to full \hyltl, i.e., beyond a single quantifier alternation.
There is precedent for such a game-based analysis of full \hyltl: 
Recently, Winter and Zimmermann showed that the existence of (Turing) computable Skolem functions for existentially quantified variables can be characterized by a game~\cite{WinterZimmermann}. 
For example, $\tsys_\all \models \phi_\inc'$ is witnessed by computable Skolem functions, as reading a prefix of length~$n$ of $\pi$ allows to compute the prefix of length~$n-1$ of $\pi'$ so that for every $\pi$, $\pi$ and the resulting $\pi'$ satisfy $\G (\propo_{\pi'} \leftrightarrow \X \propo_\pi)$.
On the other hand, $\tsys_\all \models \phi_\inc$ is not witnessed by computable Skolem functions, as the choice of the first letter of $\pi'$ depends, as explained above, on all letters of $\pi'$.
Such a Skolem function is not computable by a Turing machine, as it is not continuous.  

The game characterizing the existence of computable Skolem functions, say for a formula
\[
\phi = \forall \pi_0.\ \exists \pi_1.\ \ldots \forall \pi_{k-2}.\ \exists \pi_{k-1}.\ \psi
\]
with quantifier-free $\psi$, is a multi-player game played between a player in charge of selecting, vertex by vertex, a trace for each universally quantified variable (i.e., he has the role that Falsifier has in the previous games), and a coalition of players, one for each existentially quantified variable (i.e., the coalition has the role that Verifier has in the previous games).
Furthermore, the game must be of imperfect information in order to capture the semantics of \hyltl, where the choice of $\pi_i$ may only depend on the choice of the $\pi_j$ with $j < i$.
Thus, in the game, the player in charge of an existentially quantified $\pi_i$ only has access to the choices made so far for the $\pi_j$ for $j < i$. 
Furthermore, the game needs to incorporate a delay~\cite{HL72, KleinZimmermann,FWjournal} between the moves of the different players to capture the fact that a choice by one of the existential players may depend on future moves by the universal player (see, e.g., the formula~$\phi_\inc'$ above).
The main insight then is that a bounded delay is always sufficient, if there are computable Skolem functions at all (see, again, the difference between $\phi_\inc$ (which has no computable Skolem functions over $\tsys_\all$) and $\phi_\inc'$ (which has computable Skolem functions over $\tsys_\all$)).

\subparagraph*{Our Contribution.}
We present the first effective game-based characterization of model-checking for full \hyltl (and even \hyqptl, which allows to express all $\omega$-regular hyperproperties~\cite{Rabe16diss,FinkbeinerHHT20}), yielding a sound and complete imperfect information finite-state game with $\omega$-regular winning condition. 
This result generalizes both the alternating game with prophecies from $\forall^*\exists^*$-formulas to formulas with arbitrary quantifier prefixes and the game of Winter and Zimmermann from characterizing the existence of computable Skolem functions witnessing $\tsys \models\phi$ to characterizing $\tsys\models\phi$.

However, since $\tsys \models \phi_\inc$ holds, but does not have computable Skolem functions, the games of Winter and Zimmermann are \emph{not} a special case of the games we construct here:
Our games here are still multi-player games of imperfect information (to capture the semantics of quantification) and use prophecies (for completeness), but do not require delayed moves, as prophecies can be seen as a (restricted) form of infinite lookahead. And while the existence of computable Skolem functions is concerned with bounded lookahead, here we do indeed need infinite lookahead as witnessed by the formula~$\phi_\inc$.

Our main result shows that there is again a finite and effectively computable set of $\omega$-regular prophecies so that the coalition of players in charge of the existentially quantified variables has a winning strategy in the game with these prophecies if and only if $\tsys \models\phi$.
One challenge to overcome here is a careful definition of the prophecies, so that they are not leaking any information about choices for variables~$\pi_j$ that a player in charge of $\pi_i$ with $i < j$ must not have access to. 

Our result can also be framed in terms of Skolem function implementable by letter-to-letter transducers with ($\omega$-regular) lookahead: such a transducer computes a Skolem function for an existentially quantified variable~$\pi$ while reading values for the variables universally quantified before $\pi$ while also being able to get a regular lookahead on the trace for those universally quantified variable (much like prophecies).
The use of regular lookahead is well-studied in automata theory, see e.g.,~\cite{eilenberg1974automata,engelfriet1976top,DBLP:conf/lics/AlurFT12}.

\section{Preliminaries}

For convenience, technical terms and notations in the electronic version of this manuscript are hyper-linked to their definitions (cf.~\url{https://ctan.org/pkg/knowledge}).

Hereafter, we denote the set of nonnegative integers by $\nats$. 


\subparagraph*{Traces, Transition Systems, and Automata.}
An alphabet is a nonempty finite set. 
The sets of finite and infinite words over an alphabet~$\Sigma$ are denoted by $\Sigma^*$ and $\Sigma^\omega$, respectively. The length of finite or infinite word~$w$ is denoted by $\size{w} \in \nats \cup \set{\infty}$. 
For a word~$w$ of length at least $n$, we write $w[0,n)$ for the prefix of $w$ of length~$n$.
\AP Given $n$ infinite words~$w_0,\ldots, w_{n-1}$, let their \intro{merge} (also known as zip), which is an infinite word over $\Sigma^n$, be defined as
\[
\intro*{\combine{w_0, \ldots, w_{n-1}}}  = \vectordots{w_0(0)}{w_{n-1}(0)}\vectordots{w_0(1)}{w_{n-1}(1)}\vectordots{w_0(2)}{w_{n-1}(2)}\cdots .
\]
We define $\reintro*{\combine{w_0, \ldots, w_{n-1}}}$ for finite words~$w_0, \ldots, w_{n-1}$ of the same length analogously.

Let $\ap$ be a nonempty finite set of atomic propositions. 
\AP A \intro{trace} over $\ap$ is an infinite word over the alphabet~$\pow{\ap}$.
Given a subset~$\ap' \subseteq \ap$, the $\ap'$-projection of a trace~$t(0)t(1)t(2) \cdots$ over $\ap$ is the trace~$(t(0) \cap \ap')(t(1) \cap \ap')(t(2) \cap \ap') \cdots \in (\pow{\ap'})^\omega$.
Now, let $\ap$ and $\ap'$ be two disjoint sets, let $t$ be a "trace" over~$\ap$, and let $t'$ be a "trace" over $\ap'$. 
Then, we define $t \merge t'$ as the pointwise union of $t$ and $t'$, i.e., $t \merge t'$ is the "trace" over $\ap \cup \ap'$ defined as $(t(0) \cup t'(0))(t(1) \cup t'(1))(t(2) \cup t'(2))\cdots$.

\AP A \intro{transition system}~$\tsys = (V,E,V_\initmark, \lambda)$ consists of a finite set~$V$ of vertices, a set~$E \subseteq V \times V$ of (directed) edges, a nonempty set~$V_\initmark \subseteq V$ of initial vertices, and a labelling~$\lambda\colon V \rightarrow \pow{\ap}$ of the vertices by sets of atomic propositions.
We assume that every vertex has at least one outgoing edge.
For $v \in V$, we denote by $\Succ{v}$ the set of its successors.
A \intro{path}~$\rho$ through~$\tsys$ is an infinite sequence~$\rho = v_0v_1v_2\cdots$ of vertices with  $v_0 \in V_\initmark$ and $(v_n,v_{n+1})\in E$ for every $n \ge 0$.
\AP The \reintro[trace]{trace of $\rho$} is defined as $ \lambda(\rho ) = \lambda(v_0)\lambda(v_1)\lambda(v_2)\cdots \in (\pow{\ap})^\omega$.
\AP The \intro{set of traces} of $\tsys$ is $\intro*{\traces{\tsys}} = \set{\lambda(\rho) \mid \rho \text{ is a "path" of $\tsys$}}$.
For $V' \subseteq V$, we write $\tsys_{V'}$ to denote the transition system~$(V,E,V', \lambda)$ obtained from $\tsys$ by making $V'$ the set of initial states, and use $\tsys_v$ as shorthand for $\tsys_{\set{v}}$ for $v \in V$.

\AP A (deterministic) \intro{parity automaton}\footnote{Note that we use "parity" acceptance, as we need deterministic automata for our proofs.}~$\aut = (Q, \Sigma, q_\initmark, \delta, \col)$ consists of a finite set~$Q$ of states containing the initial state~$q_\initmark \in Q$, an alphabet~$\Sigma$, a transition function~$\delta \colon Q \times \Sigma \rightarrow Q$, and a coloring~$\col\colon Q\rightarrow \nats$ of its states by natural numbers.
Let $w = w(0) w(1) w(2) \cdots \in \Sigma^\omega$.
The run of $\aut$ on $w$ is the sequence~$q_0 q_1 q_2 \cdots $ with $q_0 = q_\initmark$ and $q_{n+1} = \delta(q_n, w(n))$ for all $n\ge 0$.
A run~$q_0q_1q_2 \cdots$ is ("parity") accepting if the maximal color appearing infinitely often in the sequence~$\col(q_0)\col(q_1)\col(q_2)\cdots $ is even.
The language ("parity") recognized by $\aut$, denoted by $L(\aut)$, is the set of infinite words over $\Sigma$ such that the run of $\aut$ on $w$ is accepting.

\begin{remark}
	Deterministic "parity automata" accept exactly the $\omega$-regular languages (see, e.g.,~\cite{gtw02} for definitions).
\end{remark}

\subparagraph*{HyperQPTL, HyperLTL and QPTL.}
\label{subsec_hyperltl}

\AP Let $\intro*{\var}$ be a countable set of \intro{trace variables}. 
The formulas of \intro*{\hyqptl} are given by the grammar
\begin{align*}
    \phi \cceq {} \exists \pi.\ \phi \mid \forall \pi.\ \phi \mid \psi 
    \qquad \psi \cceq {} \existsqp \propovar.\ \psi \mid \forallqp \propovar.\ \psi \mid \psi \mid \propo_\pi \mid \propovar \mid \lnot \psi \mid \psi \lor \psi \mid \X \psi \mid \psi \U \psi
\end{align*}
where $\propo$ and $\propovar$ range over $\ap$ and where $\pi$ ranges over $\var$.
Here, we use a tilde to decorate propositional quantifiers to distinguish them from trace quantifiers.

Note that there are two types of atomic formulas, i.e., propositions labeled by "trace variables" on which they are evaluated ($\propo_\pi$ with $\propo \in \ap$ and $\pi \in \var$) and unlabeled propositions ($\propovar \in \ap$).\footnote{We use different letters in these cases, but let us stress again that both $\propo$ and $\propovar$ are propositions in $\ap$.}
A formula is a sentence, if every occurrence of an atomic formula~$\propo_\pi$ is in the scope of a quantifier binding~$\pi$ (otherwise $\pi$ is said to be a free trace variable) and every occurrence of an atomic formula~$\propovar$ is in the scope of a quantifier binding $\propovar$ (otherwise $\propovar$ is said to be a free propositional variable).
Note that the proposition~$\propo$ in an atomic formula~$\propo_\pi$ is not considered free.
Finally, we use the usual syntactic sugar like conjunction~($\wedge$), implication~($\rightarrow$), equivalence~($\leftrightarrow$), eventually~($\F$), and always~($\G$).

\AP A (trace) \intro{variable assignment} is a partial mapping~$\Pi\colon \var \rightarrow (\pow{\ap})^\omega$.
Given a "variable"~$\pi \in \var$, and a "trace"~$t$, we denote by $\Pi[\pi \mapsto t]$ the assignment that coincides with $\Pi$ on all "variables" but $\pi$, which is mapped to $t$. 
Let $\propovar \in \ap$, let $t \in (\pow{\ap})^\omega$ be a "trace" over $\ap$, and let $t_\propovar \in (\pow{\set{\propovar}})^\omega$ be a "trace" over~$\set{\propovar}$.
We define the trace~$t[\propovar\mapsto t_\propovar] = t'{} \merge t_\propovar$, where $t'$ is the $(\ap\setminus\set{\propovar})$-projection of $t$: Intuitively, the occurrences of $\propovar$ in $t$ are replaced according to $t_\propovar$.
We lift this to sets~$T$ of "traces" by defining $T[\propovar\mapsto t_\propovar] = \set{t[\propovar\mapsto t_\propovar] \mid t \in T}$. 
Note that all traces in $T[\propovar\mapsto t_\propovar]$ have the same $\set{\propovar}$-projection, which is $t_\propovar$.

Now, for a "trace assignment"~$\Pi$, a position~$i \in \nats$, and a nonempty set~$T$ of "traces", i.e., we disregard the empty set of traces as model, we define
\begin{itemize}
    
    \item $T, \Pi, i \models \propo_\pi $ if  $\propo\in\Pi(\pi)(i)$, 
    
    \item $T, \Pi, i \models \propovar $ if for all $ t \in T$ we have $\propovar\in t(i)$,
    
    \item $T, \Pi, i \models \lnot \psi $ if  $T, \Pi, i \not\models \psi$, 
    
    \item $T, \Pi, i \models \psi_1 \lor \psi_2 $ if  $T, \Pi, i \models\psi_1$ or $T, \Pi, i \models\psi_2$, 
    
    \item $T, \Pi, i \models \X \psi $ if  $T, \Pi, i+1 \models\psi$, 
    
    \item $T, \Pi, i \models \psi_1 \U \psi_2 $ if  there is a $j \ge i$ such that $T, \Pi, j \models\psi_2$ and $T, \Pi, j' \models\psi_1$ for all $i \le j' < j$, 
        
    \item $T, \Pi, i \models \existsqp \propovar.\ \psi $ if  there exists a "trace"~$t_\propovar \in (\pow{\set{\propovar}})^\omega$ such that $ T[\propovar\mapsto t_\propovar], \Pi, i \models\psi$,   

    \item $T, \Pi, i \models \forallqp \propovar.\ \psi $ if for all "traces"~$t_\propovar \in (\pow{\{\propovar\}})^\omega$ we have $T[\propovar\mapsto t_\propovar], \Pi, i \models \psi$,
    
    \item $T, \Pi, i \models \exists \pi.\, \phi $ if  there exists a "trace"~$t \in T$ such that $T, \Pi[\pi\mapsto t], i \models \phi$, and
     
    \item $T, \Pi, i \models \forall \pi.\ \phi $ if for all "traces"~$t \in T$ we have $T, \Pi[\pi\mapsto t], i \models \phi$.

\end{itemize}

We say that an (again nonempty) set~$T$ of "traces" satisfies a sentence~$\phi$, written $T \models \phi$, if $T, \Pi_\emptyset, 0 \models \phi$ where $\Pi_\emptyset$ is the "variable assignment" with empty domain. 
We then also say that $T$ is a model of $\phi$.
A "transition system"~$\tsys$ satisfies $\phi$, written $\tsys\models\phi$, if $\traces{\tsys} \models \phi$.
Let  $\psi$ be a trace quantifier-free formula that has no free (unlabeled) propositions.
A labeled proposition of the form~$\propo_\pi$ in $\psi$ is evaluated on the traces assigned to $\pi$ and an unlabeled proposition of the form~$\propovar$ is evaluated on a trace \myquot{selected} by the quantifier binding~$\propovar$.
Hence, $T, \Pi, i \models\psi $ is independent of $T$ and we often write $\Pi \models \psi$ instead of $T, \Pi, 0 \models \psi$.

Some comments on the definition of \hyqptl are due.
\begin{remark}
We have, for technical reasons, defined \hyqptl so that every formula has a prefix of trace quantifiers followed by a formula (possibly) containing propositional quantifiers, but no more trace quantifiers.
    On the other hand, Finkbeiner et al.~\cite{FinkbeinerHHT20} require formulas to be in prenex normal form, but allow to mix trace and propositional quantification in the quantifier prefix.
    We refrained from doing so, as one can duplicate the vertices of the transition system one is interested in, so that it has enough "paths" to simulate propositional quantification by trace quantification, and can adapt the formula correspondingly. 
    This construction has a linear blowup.
\end{remark}
\begin{remark}
\AP \intro*{\hyltl} is the fragment of \hyqptl sentences that do not use the propositional quantifiers~$\existsqp$ and $\forallqp$ (and thus also not use unlabeled propositions).
\end{remark}

\AP We say that a \hyqptl formula~$\phi$ is a \intro*{\qptl} sentence if it has no trace quantifiers and no free propositional variables, i.e., all unlabeled propositional variables are in the scope of a propositional quantifier.
But a \qptl sentence may have free "trace variables", say $\pi_0, \ldots, \pi_{k-1}$ for some $k \ge 0$.
Typically, \qptl is defined without "trace variables" labelling propositions and \qptl formulas define languages over the alphabet~$\pow{\ap}$~\cite{qptl}. 
For technical necessity, we allow such labels, which implies that our formulas define languages over the alphabet~$(\pow{\ap})^k$, where $k$ is the number of "trace variables" occurring in the formula. 
More formally, a \qptl sentence with "trace variables"~$\pi_0, \pi_1, \ldots, \pi_{k-1}$ defines the language
\[
L(\phi) = \set{
\combine{t_0, \ldots, t_{k-1}} \mid t_0, \ldots, t_{k-1} \in (\pow{\ap})^\omega :
\Pi_\emptyset[\pi_0 \mapsto t_0, \ldots, \pi_{k-1} \mapsto t_{k-1}] \models \phi
}.
\]

Now, let 
$\phi = Q_0\pi_0 Q_1\pi_1 \cdots Q_{k-1}\pi_{k-1}.\ \psi$ with $Q_i \in \set{\exists,\forall}$ and trace quantifier-free $\psi$ be a \hyqptl sentence and define
$
\phi_i = Q_{i+1}\pi_{i+1} Q_{i+2}\pi_{i+2} \cdots Q_{k-1}\pi_{k-1}.\ \psi
$
for $i \in \set{0,1,\ldots,k-1}$ and $\phi_{-1} = \phi$.
Note that $\phi_{k-1} = \psi$ and that the free "trace variables" of each $\phi_i$ with $i \in \set{-1,0,\ldots,k-1}$ are exactly $\pi_0, \ldots, \pi_{i}$.
The following results follows by combining and adapting automata constructions for classical \qptl and \hyltl~\cite{ClarksonFKMRS14,FinkbeinerRS15,qptl}.
\begin{proposition}
\label{prop_automataconstruction}
\begin{enumerate}
    \item Let $\tsys$ be a "transition system".
For every $i \in \set{-1,0,\ldots,k-1}$ there is an (effectively constructible) "parity automaton"~$\aut_i^\tsys$ such that
\[
L(\aut_i^\tsys) = \{\combine{\Pi(\pi_0), \ldots, \Pi(\pi_{i})} \mid \Pi(\pi_j) \in \traces{\tsys} \text{ for $0 \leq j \leq i$ and } \traces{\tsys}, \Pi, 0 \models \phi_{i}\}.
\]

\item A language over $(\pow{\ap})^{k}$ is $\omega$-regular if and only if it is of the form~$L(\phi)$ for a \qptl-sentence~$\phi$ over some $\ap' \supseteq \ap$ with $k$ free "trace variables".\footnote{For the sake of readability, in the following, we do not distinguish between $\ap$ and $\ap'$, i.e., we assume that $\ap$ always contains enough propositions not used in our languages to quantify over.}
\end{enumerate}
\end{proposition}

\subparagraph*{Skolem Functions for HyperQPTL.}

\AP Let $\varphi = Q_0 \pi_0 \cdots Q_{k-1} \pi_{k-1}.\ \psi $ be a \hyqptl sentence such that $\psi$ is trace quantifier-free and let $T$ be a nonempty set of "traces". 
Moreover, let $i \in \set{0,1,\ldots, k-1}$ be such that $Q_i = \exists$ and let $U_i = \set{j < i \mid Q_{j} = \forall}$ be the indices of the universal quantifiers preceding $Q_i$.
Furthermore, let $f_i \colon T^{\size{U_i}} \rightarrow T$ for each such $i$ 
(note that $f_i$ is a constant, if $U_i$ is empty).
We say that a "trace assignment"~$\Pi$ with $\dom{\Pi} \supseteq \set{\pi_0, \pi_1, \ldots, \pi_{k-1}}$ is consistent with the $f_i$ if
$\Pi(\pi_i) \in T$ for all $i$ with $Q_i = \forall$ and $\Pi(\pi_i) = f_i(\Pi(\pi_{i_0}), \Pi(\pi_{i_1}), \ldots, \Pi(\pi_{i_{{\size{U_i}-1}}}))$ for all $i$ with $Q_i = \exists$, where $U_i = \set{i_0 < i_1 < \cdots < i_{\size{U_i}-1}}$.
If $\Pi \models \psi$ for each $\Pi$ that is consistent with the $f_i$, then we say that the $f_i$ are \intro{Skolem functions} witnessing $T \models \varphi$.

\begin{remark}
$T \models \varphi$ iff there are "Skolem functions" for the existentially quantified "variables" of $\varphi$ that witness $T \models \varphi$.
\end{remark}

Note that only "traces" for universal "variables" are inputs for "Skolem functions", but not those for existentially quantified "variables". 
As usual, this is not a restriction, as the inputs of a "Skolem function" for an existentially quantified "variable"~$\pi_i$ are a superset of the inputs of a "Skolem function" for another existentially quantified "variable"~$\pi_{j}$ with $j < i$.

\begin{example}[\cite{WZtracy}]
\label{example_reconstruct}
Let $\varphi = \forall \pi \exists \pi_1 \exists \pi_2.\ \G(a_\pi \leftrightarrow (a_{\pi_1} \oplus a_{\pi_2}))$.
We have $(\pow{\set{a}})^\omega \models \varphi$.
Now, for every function~$f_1 \colon (\pow{\set{a}})^\omega \rightarrow (\pow{\set{a}})^\omega$, there is a function~$f_2 \colon (\pow{\set{a}})^\omega \rightarrow (\pow{\set{a}})^\omega$ such that $f_1,f_2$ are "Skolem functions" witnessing $(\pow{\set{a}})^\omega \models \varphi$, i.e., we need to define $f_2$ such that
$(f_2(t))(n) = (f_1(t))(n)$
for all $n \in\nats$ such that $t(n) = \emptyset$ and
$(f_2(t))(n) = \flip{(f_1(t))(n)} $
for all $n \in\nats$ such that $t(n) = \set{a}$, where $\flip{\set{a}} = \emptyset$ and $\flip{\emptyset} = \set{a}$.
Hence, $f_2$ depends on $f_1$, but the value of $f_1(t)$ (for the existentially quantified~$\pi_1$) does not need to be an input to $f_2$, it can be determined from the input~$t$ for the universally quantified $\pi$. 
This is not surprising, but needs to be taken into account in our constructions.
\end{example}

\section{Gamed-based Model-Checking for HyperQPTL}
\label{sec_gamebasedmc}

In this section, we present our game-based characterization of $\tsys \models\phi$ for finite transition systems~$\tsys$ and \hyqptl sentences~$\phi$.
To this end, we introduce multi-player games with hierarchical information in \cref{subsec_prelimsGame} and then present the construction of our game in \cref{subsec_game}, while we introduce "prophecies" and prove their soundness in \cref{subsec_prophecies}.
Then, in \cref{sec_safety}, we show how to construct complete "prophecies" for the special case of "safety" formulas (to be defined formally there).
Finally, in \cref{sec_rabin}, we show how to construct complete "prophecies" for arbitrary formulas.
This allows us to first explain how to generalize the prophecy definition from $\forall^*\exists^*$ to arbitrary quantifier prefixes and then, in a second step, move from "safety" to $\omega$-regular languages. 

\subsection{Multi-player Graph Games with Hierarchical Information}
\label{subsec_prelimsGame}

We develop a game-based characterization of model-checking for \hyqptl via (multi-player) graph games with hierarchical information, using the notations of Berwanger et al.~\cite{bbb}.
First, we introduce the necessary definitions and then present our game in \cref{subsec_game}.
The games considered by Berwanger et al.\ are concurrent games (i.e., the players make their moves simultaneously), while for our purpose turn-based games (i.e., the players make their moves one after the other) are sufficient.
Turn-based games are simpler versions of concurrent games.
To avoid cumbersome notation, we introduce a turn-based variant of these games.

Fix some finite set~$C$ of players forming a coalition playing against a distinguished agent called Nature (which is \emph{not} in $C$). 
For each player~$i \in C$ we fix a finite set~$B^i$ of observations.
A game graph~$G = (V, E, v_\initmark, (\beta^i)_{i\in C})$ consists of a finite set~$V = \biguplus_{i \in C} V_i \uplus V_{\text{Nat}}$ of positions partitioned into sets controlled by some player resp.\ Nature, an edge relation~$E \subseteq V \times  V$ representing moves, an initial position~$v_\initmark \in V$, and a collection~$(\beta^i)_{i\in C}$ of observation functions~$\beta^i \colon V \rightarrow B^i$ that label, for each player, the positions with observations.
We require that $E$ has no dead-ends, i.e., for every $v \in V$ there is a $v' \in V$ with $(v,v') \in E$.

A game graph~$(V, E, v_\initmark, (\beta^i)_{i\in C})$ yields hierarchical information if there exists a total order~$\preceq$ over $C$ such that if $i \preceq j$ then for all $v,v' \in V$, $\beta^i(v) = \beta^i(v')$ implies $\beta^j(v) = \beta^j(v')$, i.e., if Player~$i$ cannot distinguish $v$ and $v'$, then neither can Player~$j$ for $i \preceq j$.

Intuitively, a play starts at position~$v_\initmark \in V$.
At position~$v$, the player that controls this position chooses a successor position~$v'$ such that $(v,v') \in E$. Now, each player~$i \in C$ receives the observation~$\beta^i(v')$ and the play continues from position~$v'$.
Thus, a play of $G$ is an infinite sequence~$v_0 v_1 v_2 \cdots$ of vertices such that $v_0 = v_\initmark$ and for all $r \ge 0$ we have $(v_r,v_{r+1}) \in E$.

A history is a prefix~$v_0 v_1 \cdots v_r$ of a play. 
We denote the set of all histories by $\hist(G)$ and extend $\beta^i\colon V \rightarrow B^i$ to plays and histories by defining $\beta^i(v_0 v_1 v_2 \cdots) = \beta^i(v_1)\beta^i(v_2)\beta^i(v_3)\cdots$. Note that the observation of the initial position is discarded for technical reasons~\cite{bbb}.
We say two histories $h$ and $h'$ are indistinguishable to Player~$i \in C$, denoted by $h \sim_i h'$, if $\beta^i(h) = \beta^i(h')$. 

A strategy for Player~$i \in C$ is a mapping~$s^i \colon V^* \rightarrow V$ that satisfies $s^i(h) = s^i(h')$ for all histories~$h,h'$ with $h \sim_i h'$ (i.e., the move selected by the strategy only depends on the observations of the history).
A play~$v_0 v_1 v_2 \cdots$ is consistent with $s^i$ if for every $r \ge 0$, we have $v_{r+1} = s^i(v_0 v_1 \cdots v_r)$.
A play is consistent with a collection of strategies~$(s^i )_{i \in C}$ if it is consistent with each $s^i$.
The set of possible outcomes of a collection of strategies is the set of all plays that are consistent with it.
As usual, a strategy is finite-state, if it is implemented by some Moore machine.


Lastly, a game $\mathcal{G}$ consists of a game graph $G$ and a winning condition~$W \subseteq V^\omega$, where $V$ is the set of positions of $G$.
A play is winning if it is in $W$ and a collection of strategies is winning if all its outcomes are winning.

\begin{proposition}[\cite{DBLP:conf/focs/PnueliR90,bbb}]
\label{prop_distributedgames}
\begin{enumerate}
    \item \label{prop_distributedgames_decidability}
        It is decidable, given a game with hierarchical information with $\omega$-regular winning condition, whether it has a winning collection of strategies.
    \item \label{prop_distributedgames_finitestate}
        A game with hierarchical information with $\omega$-regular winning condition has a winning collection of strategies if and only if it has a winning collection of finite-state strategies.
\end{enumerate}
\end{proposition}

\subsection{The Model-checking Game}
\label{subsec_game}

\AP For the remainder of this paper, we fix a \hyqptl sentence~$\intro*{\myphi}$ and a "transition system"~$\intro*{\mytsys}$.
We assume (w.l.o.g.)\footnote{\label{fn_nf}The following reasoning can easily be extended to general sentences with arbitrary quantifier prefixes, albeit at the cost of more complex notation. We substantiate this claim in \cref{remark_nonstrictalternation} on \cpageref{remark_nonstrictalternation}.}
\[
\reintro[myphi]{\phi} = \reintro[myphi]{\forall \pi_0 \exists \pi_1 \cdots \forall \pi_{k-2} \exists \pi_{k-1}.\ \psi}
\]
such that $\intro*{\mypsi}$ is trace quantifier-free, define 
\[
\reintro[myphi]{\phi_i} = \reintro[myphi]{Q_{i+1} \pi_{i+1} Q_{i+2} \pi_{i+2} \cdots \forall \pi_{k-2} \exists \pi_{k-1}.\ \psi}
\]
for $i \in \set{-1, 0, \ldots, \mathmbox{k-1}}$
and use the automata~$\intro*{\auti{i}}$ constructed in \cref{prop_automataconstruction} satisfying 
\[
\AP 
 L(\auti{i}) =   \{ \combine{\Pi(\pi_0), \ldots, \Pi(\pi_{i})} \mid \Pi(\pi_j) \in \traces{\mytsys} \text{ for all $ 0 \le j \leq i$ and }\traces{\mytsys}, \Pi, 0 \models \myphi_{i} \}.
\]
We use the notation $\reintro*{\auti{i}[q]}$ to denote the "parity automaton" obtained from $\auti{i}$ by making its state~$q$ the initial state.
Finally, let $\reintro*{\auti{k-1}} = \reintro[aut]{(Q,\Sigma,q_\initmark,\delta,\col)}$. 
Note that $\auti{k-1}$ accepts the "trace assignments" coming from "paths" trough $\mytsys$ that satisfy~$\mypsi$.
We say that  $\auti{k-1}$ checks $\mypsi$.

\AP We define a multi-player game~$\intro*{\game{\mytsys, \myphi}}$ with hierarchical information induced by the "transition system" $\mytsys$ and the \hyqptl sentence~$\myphi$.
This game is played between \intro{Falsifier} (who takes on the role of Nature, cf.\ \cref{subsec_prelimsGame}), who is in charge of providing "traces" (from "paths" trough the "transition system") for the universally quantified "variables", and a coalition of \reintro{Verifier-players} $\set{1,3,\ldots,k-1}$ (\reintro*{\vplayer{i}} for short), who are in charge of providing "traces" (also from "paths" trough the "transition system") for the existentially quantified "variables".
The goal of "Falsifier" is to prove $\mytsys \not\models \myphi$ and the goal of the coalition of "Verifier-players" is to prove $\mytsys \models \myphi$.
Therefore, the $k$ "traces" built during a play are read synchronously by the "parity automaton" $\auti{k-1}$ accepting the "trace assignments" that satisfy~$\mypsi$ (recall that $\mypsi$ is the trace quantifier-free part of $\myphi$).

The game has two phases,
an initialization phase where initial vertices for all "paths" through the "transition system" are picked, and a second phase (of infinite duration) where the "paths" (which induce the "trace assignment") are built. 
Formally, in the initialization phase, a position of the game is of the form $((v_0,\ldots,v_{i-1},\underbrace{\bullet,\ldots,\bullet}_{k-i \text{ times}}),q_\initmark,i)$ where $v_0,\ldots,v_{i-1} \in V_\initmark$ are initial vertices in the "transition system"~$\mytsys$, $\bullet$ is a fresh (placeholder) symbol, $q_\initmark \in Q$ is the initial state of $\auti{k-1}$, and $i \in \set{0,1,\cdots,\mathmbox{k-1}}$.
In the second phase, a position of the game is of the form $((v_0,\ldots,v_{k-1}),q,i)$ where $v_0,v_1,\ldots,v_{k-1} \in V$, $q \in Q$, and $i \in \set{0,1,\ldots,\mathmbox{k-1}}$.
A vertex whose last component is an odd~$i$ is controlled by \vplayer{i} and a vertex whose last component is even is controlled by "Falsifier".

The edges (also called moves) of the game graph are defined as follows, where the first two items are the moves in the initialization phase:
\begin{itemize}

    \item $\bigl(((v_0,\ldots,v_{i-1},\bullet,\ldots,
    \bullet),q_\initmark,i),((v_0,\ldots,v_{i-1},v_i,\bullet,\ldots,
    \bullet),q_\initmark,i+1)\bigr)$ for all $v_i \in V_\initmark$ and all $i \in \set{0,1,\ldots,\mathmbox{k-2}}$: An initial vertex in $\mytsys$ for the $i$-th "path" is picked.

    \item $\bigl(((v_0,\ldots,v_{k-2},\bullet),q_\initmark,k-1),((v_0,\ldots,v_{k-2},v_{k-1}),q,0)\bigr)$ for all $v_{k-1} \in V_\initmark$ where \[q = \delta(q_\initmark,\combine{\lambda(v_0),\ldots,\lambda(v_{k-1})}):\] An initial vertex in $\mytsys$ for the last "path" is picked.
    With this move, the initialization phase is over and the state of $\auti{k-1}$ checking $\mypsi$ is updated for the first time.

    \item $\bigl(((v'_0,\ldots,v'_{i-1},v_i,v_{i+1},\ldots,
    v_{k-1}),q,i),((v'_0,\ldots,v'_{i-1},v'_i,v_{i+1},\ldots,
    v_{k-1}),q,i+1)\bigr)$ for all $(v_i,v'_i) \in E$ and all $i \in \set{0,1,\ldots,\mathmbox{k-2}}$: The $i$-th "path" is updated by moving to a successor vertex in $\mytsys$.
    
    \item $\bigl(((v'_0,\ldots,v'_{k-2},
    v_{k-1}),q,k-1),((v'_0,\ldots,,v'_{k-2},v'_{k-1}),q',0)\bigr)$ for all $(v_{k-1},v'_{k-1}) \in E$ where $q' = \delta(q,\combine{\lambda(v'_0),\lambda(v'_1),\ldots,\lambda(v'_{k-1})})$: The last "path" is updated by moving to a successor vertex in $\mytsys$. Simultaneously, the state of $\auti{k-1}$ checking $\mypsi$ is updated.
    
\end{itemize}
The initial position is $((\bullet,\ldots,\bullet),q_\initmark,0)$.
We use a "parity" winning condition. 
The color of all positions of the initialization phase is $0$ (the color is of no consequence as these vertices are seen only once during the course of a play).
The color of positions of the form~$((v_0,\ldots,v_{k-1}),q,i)$ is the color that $q$ has in $\auti{k-1}$, i.e., a play is winning for the "Verifier-players" if and only if the trace assignment picked by them and "Falsifier" satisfies~$\psi$.

The game described must be a game of hierarchical information to capture the fact that the "Skolem function" for an existentially quantified $\pi_i$ depends only on the universally quantified "variables"~$\pi_{j}$ with $j \in \set{0,2,\ldots, i-1}$.
\AP To capture that, we define an equivalence relation~$\intro*{\equivrel{i}}$ between positions of the game for $i \in \set{1,3,\ldots,\mathmbox{k-1}}$ which ensures that for \vplayer{i}, two positions are indistinguishable if they coincide on their first $i+1$ components and additionally belong to the same player.
Formally, regarding positions from the initialization phase, let $((v_0,\ldots,v_{m-1},\bullet,\ldots,\bullet),q_\initmark,m)$ and $((v'_0,\ldots,v'_{n-1},\bullet,\ldots,\bullet),q_\initmark,n)$ be $\equivrel{i}$-equivalent if $v_j = v'_j$ for all $j \le i$ and $m = n$.
Regarding all other positions, let $((v_0,\ldots,v_{k-1}),p,m)$ and $((v'_0,\ldots,v'_{k-1}),q,n)$ be $\equivrel{i}$-equivalent if $v_j = v'_j$ for all $j \le i$ and $m = n$.
Now, we define the observation functions: For \vplayer{i}, it maps positions to their $\equivrel{i}$-equivalence classes.

\subsection{The Model-Checking Game with Prophecies}
\label{subsec_prophecies}

The game~$\game{\mytsys, \myphi}$ as introduced in \cref{subsec_game} does not capture $\mytsys \models \myphi$. While it is sound (see \cref{thm_soundprophy} for the empty set of prophecies), i.e., the coalition of "Verifier-players" having a winning collection of strategies for $\game{\mytsys, \myphi}$ implies that $\mytsys \models \myphi$, but the converse is not necessarily true. 
This is witnessed, e.g., by a transition system~$\tsys$ with $
\traces{\tsys} = (\pow{\set{\propo}})^\omega$ and the sentence~$\forall \pi.\ \exists \pi'.\ \propo_{\pi'}\leftrightarrow \F\propo_\pi$.
As explained in the introduction, the "Verifier-player" does not have a strategy to select, step-by-step, a "trace"~$t'$ for $\pi'$ while given, again step-by-step, a trace~$t$ for $\pi$, as the choice of the first letter of $t'$ depends on all positions of $t$.
Hence, the coalition does not win~$\game{\tsys, \phi}$, even though~$\tsys \models \phi$.

\AP In the following, we show how \intro{prophecies}, binding commitments about future moves by "Falsifier", make the game-based approach to model-checking complete. 
In our example, "Falsifier" has to make, with his first move, a commitment about whether he will ever play a $\propo$ or not. This allows the "Verifier-player" to pick the \myquot{right} first letter of $t'$ and thereby win the game, if "Falsifier" honors the commitment. If not, the rules of the game make him loose. 
To add "prophecies" to $\game{\mytsys, \myphi}$, we do not change the rules of the game, but instead modify the transition system the game is played on (to allow "Falsifier" to select truth values for the \intro{prophecy variables}\footnote{Note that "prophecy variables" are Boolean variables that are set by "Falsifier" during each move of a play and should not be confused with "trace variables" or with propositional variables.}, which is the mechanism he uses to make the commitments) and modify the formula (to ensure that "Falsifier" loses if he breaks a commitment). 
Hence, given $\mytsys$ and $\myphi$, we construct~$\tsysmanip$ and $\phimanip$ such that the following two properties are satisfied:
\begin{itemize}
    \item Soundness: If the coalition of "Verifier-players" has a winning collection of strategies for $\game{\tsysmanip, \phimanip}$, then $\mytsys \models \myphi$.
    \item Completeness: If $\mytsys \models \myphi$, then the coalition of "Verifier-players" has a winning collection of strategies for $\game{\tsysmanip, \phimanip}$.
\end{itemize}
The challenge here is to construct the \myquot{right} "prophecies" that allow us to prove completeness, as soundness is independent of the "prophecies" chosen.

Recall that $\pi_0, \pi_2, \ldots, \pi_{k-2}$ resp.\ $\pi_1, \pi_3, \ldots, \pi_{k-1}$ are the universally resp.\ existentially quantified "variables" in our fixed formula~$\myphi$. 
We frequently need to refer to their indices.
\AP Hence, we define $\intro*{\ieven} = \reintro[index]{\set{0,2,\ldots, k-2}}$ and $\reintro*{\iodd} = \reintro[index]{\set{1,3,\ldots, k-1}}$.

We begin by defining the transition system~$\tsysmanip$ from the fixed transition system~$\mytsys$, in which "paths" additionally determine truth values for "prophecy variables".

\AP \begin{definition}[\intro{System manipulation}]\label{def_systemmanip}
Let~$\calP = (\pp_i)_{i \in \ieven}$ be a collection of sets of atomic propositions such that $\ap$, the $\pp_i$, and $\set{\markprop_i \mid i\in \ieven}$ are all pairwise disjoint, where the $\markprop_i$ are propositions used to mark copies of the transition system~$\mytsys = (V, E, V_\initmark, \lambda)$.

For $i \in \ieven$, we define
$\mytsys^{\pp_i} = (V^{\pp_i},E^{\pp_i},V_\initmark^{\pp_i}, \lambda^{\pp_i})$ over $\ap \uplus \pp_i \uplus \set{\markprop_i}$ where $V^{\pp_i} = V \times \pow{\pp_i} \times \set{i}$, $E^{\pp_i} = \set{ ((s,A,i),(s',A',i)) \mid (s,s') \in E \text{ and } A,A' \in \pow{\pp_i}}$, $V_\initmark^{\pp_i} = V_\initmark \times \pow{\pp_i} \times \set{i}$ and $\lambda^{\pp_i}(s,A,i) = \lambda(s) \cup A \cup \set{\markprop_i}$.

Furthermore, we define $\mytsys^{\calP} = (V^{\calP}, E^{\calP}, V_\initmark^{\calP}, \lambda^{\calP})$ as the disjoint union of the $\mytsys^{\pp_i}$, where a vertex of~$\mytsys^{\calP}$ is in $V_\initmark^{\calP}$ if and only if it is in some $V_\initmark^{\pp_i}$. 
Consequently, we have $
\traces{\mytsys^{\calP}} = \bigcup_{i \in \ieven} \set{ {t \merge {t'} \merge \set{\markprop_i}^\omega} \mid 
t \in \traces{\mytsys} \text{ and }
 t' \in 
 {(\pow{{\pp_i}})}^\omega 
 }
$.
\end{definition}
An effect of this definition is that all "paths" trough the manipulated  system select truth values for the "prophecy variables" with each move. 
Furthermore, each such "path" is marked by a proposition of the form~$\markprop_i$ indicating for which $i$ the "prophecy variables" are selected.
This will later be used to ensure that "Falsifier" selects the correct "prophecies" (cf.\ \cref{def_propmani}) for each "path" he picks for a universally quantified variable.
For the "Verifier-players", this additional information is simply ignored, as can be seen from the manipulated formula (cf.\ \cref{def_propmani}) we introduce now.

We define the sentence~$\phimanip$ which ensures that "Falsifier" loses, if he breaks his commitments made via "prophecy variables":
for every "prophecy variable", we have an associated \reintro{prophecy}, a language~$\prophy \subseteq ((\pow{\ap})^i)^\omega$ for some $i$.
Note that "Falsifier" can only make commitments about his own moves, as the "Verifier-players" could otherwise falsify the commitments (made by "Falsifier") about their moves: "Prophecies" can only refer to "traces" for universally quantified "variables".
Also, we will have "prophecies" for each universally quantified "variable".

\AP \begin{definition}[\intro{Property manipulation}]\label{def_propmani}
Let $\Xi = (\Xi_i)_{i \in \ieven}$ be a family~$\Xi_i = \set{\xi_{i,1}, \ldots, \xi_{i,n_i}}$ of sets of \qptl sentences (over $\ap$) such that each $\xi \in \Xi_i$ uses only "trace variables"~$\pi_j$ with even $j \leq i$.
Let $\calP = (\pp_i)_{i \in \ieven} $ satisfy the disjointness requirements of \cref{def_systemmanip} and $\pp_i = \set{\propo_{i,1}, \ldots, \propo_{i,n_i}}$, i.e., we have $\size{\Xi_i} = \size{\pp_i}$.
We define the \hyqptl sentence~$\reintro*{\phimanip}$ as
\reintro[property manipulation]{
\[
\forall \pi_0 \exists \pi_1 \cdots \forall \pi_{k-2} \exists \pi_{k-1}.\ \left[ \bigwedge\nolimits_{i \in \ieven}
(\markprop_i)_{\pi_i} \wedge
 \G \left( \bigwedge\nolimits_{\ell=1}^{n_i} (({\propo_{i,\ell}})_{\pi_{i}} \leftrightarrow \xi_{i,\ell}) \right) \right] \rightarrow \psi.
\]
}
We denote the trace quantifier-free part of $\phimanip$ by $\intro*{\psimanip}$.
\end{definition}
Note that $\phimanip$ has the same quantifier prefix as $\myphi$ and that restricting each $\xi_{i,\ell}$ to "trace variables"~$\pi_j$ with even $j \le i$ ensures that the "prophecies" only refer to "trace variables" under the control of "Falsifier".
Also note that the truth values of "prophecy variables" on "trace variables" under the control of the "Verifier-players" are not used in the formula.
Finally, "Falsifier" has to pick the $i$-th path in $\mytsys^{\pp_i}$ (and thus select valuations for the "prophecy variables" in $\pp_i$), otherwise he loses immediately. 

Our construction is sound, independently of the choice of "prophecies".

\begin{lemma}[name={},restate=Soundprophy]\label{thm_soundprophy}
Let $\Xi$ and $\calP$ satisfy the requirements of \cref{def_propmani}.
If the coalition of "Verifier-players" has a winning collection of strategies in $\game{\tsysmanip, \phimanip}$, then $\mytsys \models \myphi$.
\end{lemma}

\begin{proof}
Let the coalition of "Verifier-players" have a winning collection of strategies in $\game{\tsysmanip, \phimanip}$. 
We construct "Skolem functions"~$f_\pi$ for the "variables"~$\pi$ existentially quantified in $\myphi$ witnessing~$\mytsys \models \myphi$.

Let $t_0, t_2, \ldots t_{k-2}$ be "traces" of $\mytsys$ for the universally quantified "variables".
For each such $t_i$, we fix a path~$\rho_i$ through~$\tsysmanip$ such that the (point-wise) $\ap$-projection of the "trace" of $\rho_i$ is $t_i$, such that $\markprop_i$ holds on every position of $\rho_i$, and such that the "prophecy variables"~$\propo_{i,\ell}$ are picked correctly, 
i.e., such that $\propo_{i,\ell}$ is satisfied at position~$n$ if and only if the suffixes of $t_0, t_2, \ldots, t_{i}$ starting at position~$n$ satisfy $\xi_{i,\ell}$.

Consider the play of $\game{\tsysmanip, \phimanip}$ where "Falsifier" plays such that he constructs the $\rho_i$ and the "Verifier-players" use their winning collection of strategies to construct "paths"~$\rho_1, \rho_3, \ldots, \rho_{k-1}$ through $\tsysmanip$. 

Recall that $\lambda^{\calP}$ is the labelling function of $\tsysmanip$.
As the collection of strategies is winning, 
\[
\combine{\lambda^{\calP}(\rho_0),\lambda^{\calP}(\rho_1),\ldots,\lambda^{\calP}(\rho_{k-1})}
\]
is accepted by the automaton constructed from~$\phimanip$, i.e., the "trace assignment"~$\Pi$ mapping each $\pi_i$ to $\lambda^{\calP}(\rho_i)$ satisfies 
\[
\left[ \bigwedge\nolimits_{i \in \ieven}
(\markprop_i)_{\pi_i} \wedge
 \G \left( \bigwedge\nolimits_{\ell=1}^{n_i} (({\propo_{i,\ell}})_{\pi_{i}} \leftrightarrow \xi_{i,\ell}) \right) \right] \rightarrow \mypsi.
\]
By construction of the $\rho_i$ for $i \in \ieven$, $\Pi$ satisfies the premise 
\[
 \bigwedge\nolimits_{i \in \ieven}
(\markprop_i)_{\pi_i} \wedge
 \G \left( \bigwedge\nolimits_{\ell=1}^{n_i} (({\propo_{\ell,i}})_{\pi_{i}} \leftrightarrow \xi_{\ell,i}) \right) 
.\]
Hence, $\Pi$ must satisfy $\mypsi$.

Finally, by the definition of the hierarchical information in the game~$\game{\tsysmanip, \phimanip}$, $\rho_i$ for even $i$ only depends on the "paths"~$\rho_0, \rho_1, \ldots, \rho_{i-1}$, but not on the "paths"~$\rho_{i+1}, \rho_{i+2}, \ldots, \rho_{k-1}$ (as they are hidden to the $i$-th "Verifier-player").
Thus, we can inductively, for $i \in \iodd$, define $f_{\pi_i}(t_0, t_2, \ldots, t_{i-1}) $ as the $\ap$-projection of $\lambda^{\calP}(\rho_i)$.

Then, the functions~$f_\pi$ just defined are indeed "Skolem functions" witnessing $\mytsys \models \myphi$.
\end{proof}

\section{Complete Prophecies for Safety Properties}
\label{sec_safety}

We show that there are sets $\Xi$ and $\calP$ such that if $\mytsys \models \myphi$, then the coalition of "Verifier-players" wins $\game{\tsysmanip,\phimanip}$.
We first consider the case where $\mypsi$ is a "safety" property, i.e., the automaton~$\auti{k-1}$, which checks $\mypsi$, is a deterministic "safety automaton".
\AP A \intro{safety automaton} is a "parity automaton" using only two colors, an even one and an odd one, and moreover, all its states are even-colored except for an odd-colored sink state.
In other words, all runs avoiding the single unsafe state are accepting.
We start with "safety" as this allows us to work with \myquot{simpler} "prophecies" while presenting all underlying concepts needed for arbitrary quantifier prefixes. The idea for "safety" is that a "prophecy" should indicate which successor vertices are safe for the "Verifier-players" to move to, i.e., from which successor vertices it is possible to successfully continue the play without immediately losing.
In the general case, we have to additionally handle a more complex acceptance condition. 
This is shown in \cref{sec_rabin}.

\begin{definition}[Prophecy construction for safety]
\AP 
For each $i \in \ieven$, each vector $\bar{v} = (v_0,v_1,\ldots,v_{i+1})$ of vertices of $\mytsys$, and each state $q$ of $\auti{i+1}$ we define
\begin{equation*}
\begin{split}
     & \intro*{\safeprophy{i}{q}{\bar{v}}} = \{ \combine{t_0,t_2,\ldots,t_{i}} \mid t_0 \in \traces{\mytsys_{v_0}}, t_2 \in \traces{\mytsys_{v_2}},\ldots,t_{i} \in \traces{\mytsys_{v_{i}}}, \text{ and there are } \\ 
     & \quad  t_{1} \in \traces{\mytsys_{v_1}}, t_{3} \in \traces{\mytsys_{v_3}},\ldots, t_{i+1} \in \traces{\mytsys_{v_{i+1}}} \text{ s.t.\ } \combine{t_0,\ldots,t_{i+1}} \in L(\auti{i+1}[q]) \}. 
\end{split}
\end{equation*}
\end{definition}

Note that the $t_i$ for $i \in \iodd$ may depend on all $t_j$ with $j \in \ieven$.
Our game definition ensures that the choice for an existential variable~$\pi_i$ only depends on the choices for $\pi_j$ with~$j< i$.

Also note that each "prophecy" is an $\omega$-regular language, as $\omega$-regular languages are closed under projection, the analogue of existential quantification. 
So, due to \cref{prop_automataconstruction}, there are \qptl-sentences expressing the "prophecies" allowing us to verify them in the formula~$\phimanip$.

 Next, we define the sets~$\Xi = (\Xi_i)_{i \in \ieven}$ and $\calP = (\pp_i)_{i \in \ieven}$ of \qptl formulas expressing the "prophecies" defined above and the corresponding "prophecy variables".
\begin{definition}
\label{def_safetymanip}
 Let $\Xi_i$ be the set of \qptl formulas that contains sentences~$\pformula{i}{q}{\bar{v}}$ for each state~$q$ of $\auti{i+1}$ and each vector $\bar{v}$ of vertices expressing the "prophecy"~$\safeprophy{i}{q}{\bar{v}}$ using only "trace variables" $\pi_j$ with even $j \leq i$.
 \AP Let $\pp_i$ be the set that contains the "prophecy variable" \intro*{$\psafe{i}{q}{\bar{v}}$} for each state $q$ of $\auti{i+1}$ and each vector $\bar{v}$ of vertices. 
 The "prophecy variable" $\psafe{i}{q}{\bar{v}}$ corresponds to $\pformula{i}{q}{\bar{v}}$.
\end{definition}

Recall that $\mypsi$ is the maximal trace quantifier-free subformula of $\myphi$.
Our main technical lemma shows that the "prophecies" defined above are indeed complete.

\begin{lemma}[name={},restate=Safetycompleteness]
\label{thm_safetycompleteness}
 Let $\mytsys \models \myphi$ such that $\mypsi$ is a "safety" property.
 Then the coalition of "Verifier-players" has a winning collection of strategies in $\game{\tsysmanip, \phimanip}$ with $\Xi$ and $\calP$ as in \cref{def_safetymanip}.
\end{lemma}

Before we turn to the proof of \cref{thm_safetycompleteness}, we introduce some notation about the game~$\game{\tsysmanip, \phimanip}$ and define the strategies we will prove winning.

Firstly, recall that $\tsysmanip$ is the union $\biguplus_{i \in \ieven} \mytsys^{P_i}$.
Thus, a vertex~$v$ of $\tsysmanip$ is of the form~$(s,A,i) \in V \times \pow{\pp_{i}} \times \set{i}$ for some $\pp_i$, where $V$ is the set of vertices of $\mytsys$, and the label~$\lambda^{\pp_i}(v)$ of $v$  is $\lambda(s) \cup A \cup \set{\markprop_i}$, i.e., the union of its original (meaning $\ap$-based) label~$\lambda(s)$ in $\mytsys$, its associated "prophecy variables"~$A$, and its marker~$\markprop_i$ indicating that $v$ belongs to $\mytsys^{\pp_{i}}$.
\AP We need to access these bits of information. Hence, we define $\intro*{\orig{v}} = \lambda(s)$, $\intro*{\prophecies{v}} = A$, and $\intro*{\marking{v}} = \markprop_i$.
Given a "path"~$\rho = v_0v_1v_2\cdots$ through $\tsysmanip$, let $\reintro*{\orig{\rho}} = \orig{v_0}\orig{v_1}\orig{v_2}\cdots$.
Moreover, for $v = (s,A,i)$ in $\mytsys^{\pp_{i}}$, we are often interested in reasoning about $s$ in $\mytsys$. To simplify our notation, we will also write $v$ for the (unique) vertex $s$ in $\mytsys$ induced by $v$. 

Secondly, recall that $\psimanip$ is the trace quantifier-free part of $\phimanip$.
Let $\autb$ denote the "parity automaton" $(Q_\autb,\Sigma,q_\initmark^\autb,\delta_\autb,\col_\autb)$ that accepts the "trace assignments" that satisfy $\psimanip$.

Lastly, recall that a position of $\game{\tsysmanip, \phimanip}$ is of the form $((v'_0,\ldots,v'_{i-1},v_i,\ldots, v_{k-1}),q,i)$ where $v'_0,\ldots,v'_{i-1},v_{i},\ldots,v_{k-1}$ are vertices in the "transition system"~$\tsysmanip$, $q \in Q_\autb$ is a state of the "parity automaton" $\autb$ checking $\psimanip$, and $i \in \set{0,1,\ldots,\mathmbox{k-1}}$.
Technically, in the initialization phase, a position is of the form $((v_0,\ldots,v_{i-1},\underbrace{\bullet,\ldots,\bullet}_{k-i \text{ times}}),q_\initmark,i)$ where $v_0,\ldots,v_{i-1}$ are vertices and $\bullet$ is a placeholder.
For simplicity, we always refer to a position with $((v'_0,\ldots,v'_{i-1},v_i,\ldots, v_{k-1}),q,i)$, even though $\bullet$ might be present.
Also, for convenience, we define the successors $\Succ{\bullet}$ of the placeholder symbol~$\bullet$ to be $V_\initmark$, the set of initial vertices of $\mytsys$.

\AP We say a play in $\game{\tsysmanip, \phimanip}$ is in \intro*{\round{r}[i]} for $r \geq 0$ and $i \in \set{0,1,\ldots,k-1}$ if the play is in a position of the form $((v'_0,\ldots,v'_{i-1},v_i,\ldots, v_{k-1}),q,i)$ for the $(r+1)$-th time.
Also, we say we are in \reintro*{\round{r}} if the play is in \round{r}[i] for some $i$.

Given a play $\alpha$ in $\game{\tsysmanip, \phimanip}$, for $i \in \set{0,1,\ldots,k-1}$, let $\rho_i^{\alpha}$ denote the "path" through $\tsysmanip$ that is induced by $\alpha$ when considering the moves made in \round{r}[i] for $r \geq 0$:
Formally, if in \round{r}[i] the move from position $p = ((v'_0,\ldots,v'_{i-1},v_i,\ldots,v_{k-1}),q,i)$ to $p' = ((v'_0,\ldots,v'_{i-1},v'_i,\ldots, v_{k-1}),q',(i+1) \mod k)$ is made, then $\rho_i^{\alpha}(r) = v'_i$.
Note that $q' = q$, unless $i = k-1$, then \[q' = \delta_{\autb}(q,\combine{\lambda^{\calP}(v'_0),\ldots,\lambda^{\calP}(v'_{k-1})}).\]
Furthermore, note that the move $(p,p')$ is uniquely identified by $p$ and $v'_i$.
So in the future, we simply write the move from $v_i$ to $v'_i$ when $p$ is clear from the context.
We drop the index $\alpha$ from $\rho_i^{\alpha}$ and write~$\rho_i$ when $\alpha$ is clear from the context.

We continue by defining a collection of strategies for the "Verifier-players" and then show that $\mytsys \models \myphi$ implies that this collection is winning in $\game{\tsysmanip, \phimanip}$.
Recall that $\phimanip$ is $\forall \pi_0 \exists \pi_1 \cdots \forall \pi_{k-2} \exists \pi_{k-1}.\ \psimanip$.
Assume the play is in \round{r}[i] for $r \geq 0$ and $i \in \iodd$.
Thus, it is in a position of the form
\[
((v'_0,\ldots,v'_{i-1},v_i,\ldots,v_{k-1}),q,i)
\]
and \vplayer{i} has to move.
We review some information \vplayer{i} can use to base her choice on.
In \round{r}, \vplayer{i} has access to the first $r$ letters of the "traces" $\orig{\rho_0},\orig{\rho_1},\ldots,\orig{\rho_i}$ induced by the play.\footnote{\label{foot_strategy}For $j < i$, she actually has access to $r+1$ letters, but our construction is independent of the last ones.}
Let $q_{i}$ be the state that $\auti{i}$ has reached on $\combine{\orig{\rho_0}[0,r),\orig{\rho_1}[0,r),\ldots,\orig{\rho_i}[0,r)}$ for $i \in \iodd$.

\begin{definition}[Strategy definition for safety]
\label{def_strategysafety}
\AP
Let
\[\possmoves{i}{v_0',\ldots,v'_{i-1}, v_i, q_{i}} = \set{ v'_i \in \Succ{v_i} \mid \psafe{i-1}{q_{i}}{\bar{v}} \in \prophecies{v'_{i-1}} }\]
where $\bar{v} = (v_0',\ldots,v'_i)$.
If $\possmoves{i}{v_0',\ldots,v'_{i-1}, v_i, q_{i}}$ is nonempty, then \vplayer{i} can move to any $v'_i$ in the set (the \myquot{\nonemptycase-case}).
If it is empty, then \vplayer{i} can move to any $v'_i \in \Succ{v_i}$ (the \myquot{\emptycase-case}).
\end{definition}

\begin{proof}[Proof of \cref{thm_safetycompleteness}.]
Let us prove that the strategies defined above form a winning collection.
To this end, let $\alpha$ be a play in $\game{\tsysmanip, \phimanip}$ that is consistent with the strategies and satisfies the following assumption.

\begin{assumption}
\label{assump_strategy}
Let $\rho_0^\alpha,\ldots,\rho_{k-1}^\alpha$ be the paths induced by $\alpha$ and let $\Pi$ be the "trace assignment" mapping $\pi_i$ to $\lambda^{\calP}(\rho_i)$ for $i \in \ieven$. We assume that $\Pi$ satisfies the premise of $\psimanip$.
\end{assumption}

Intuitively, this assumption is satisfied by "Falsifier", if he picks "paths" that start in the intended sub-parts of the manipulated "transition system", and furthermore, he always truthfully indicates which "prophecies" hold.
Plays in which the assumption is violated are trivially won by the "Verifier-players".
Thus, in the following we can focus on $\alpha$'s satisfying the assumption.

We show by induction over \round{r}[i] for $r \ge 0$ and $i \in \iodd$ that under this assumption, \vplayer{i} used the \nonemptycase-case in this round to pick her move. 
To proceed by nested induction, i.e., with an outer induction going from \round{r}[0] to \round{r+1}[0] for each $r \ge 0$, and an inner induction going from \round{r}[i] to \round{r}[i+2] for fixed $r \ge 0$ and $i \in \iodd \setminus \set{k-1}$.

We begin with the outer induction.
At the beginning of \round{r}[0], the play~$\alpha$ is in a position of the form~$((v_0,\ldots,v_{k-1}),q,0)$.
Let $q_{k-1}$ be the state that $\auti{k-1}$ has reached on 
\[
\combine{\orig{\rho_0}[0,r),\orig{\rho_1}[0,r),\ldots,\orig{\rho_{k-1}}[0,r)}.
\]
We prove by induction on $r$ the following auxiliary statement:
\begin{align}
\label{eq_safe}
 & \forall {t}_0 \in \traces{\mytsys_{\Succ{v_0}}}\  \exists {t}_1 \in \traces{\mytsys_{\Succ{v_1}}} \cdots \forall {t}_{k-2} \in \traces{\mytsys_{\Succ{v_{k-2}}}}\  \exists {t}_{k-1} \in \traces{\mytsys_{\Succ{v_{k-1}}}} :\notag\\
 & \qquad\combine{{t}_0,\ldots,{t}_{k-1}} \in L(\auti{k-1}[q_{k-1}]). 
\end{align}

\subparagraph*{Outer Base Case, \boldmath\round{0}[0].}
The starting position of the game, which is the position in \round{0}[0], is $((\bullet,\ldots,\bullet),q_\initmark^\autb,0)$. As $\Succ{\bullet} = V_\initmark$, $\mytsys_{V_\initmark} = \mytsys$, and $\auti{k-1}$ has reached $q_\initmark$ on the empty prefix, \cref{eq_safe} is equivalent to $\tsys \models \myphi$, which is true by assumption of the statement of \cref{thm_safetycompleteness}.

\subparagraph*{Outer Inductive Step, \boldmath\round{r}[0] to \boldmath\round{r+1}[0].}
To show \cref{eq_safe} for \mbox{\round{r+1}[0]}, we inductively go over the rounds \round{r}[i] for each \vplayer{i}.

In \round{r}[i], as mentioned in the strategy definition (cf.\ \cref{def_strategysafety}), \vplayer{i} has access to the first $r$ letters of the "traces" $\orig{\rho_0},\orig{\rho_1},\ldots,\orig{\rho_i}$ induced by the play prefix.
Let $q_{i}$ be the state that $\auti{i}$ has reached on 
\[
\combine{\orig{\rho_0}[0,r),\orig{\rho_1}[0,r),\ldots,\orig{\rho_{j}}[0,r)}
\]
for $i \in \iodd$.
In \round{r}[i], the play is in a position $((v'_0,\ldots,v'_{i-1},v_i,\ldots,v_{k-1}),q,i)$.
We prove by induction on $i \in \iodd$ another auxiliary statement and our original claim: there exists some $ v'_{i} \in \Succ{v_i}$ such that
\begin{align}
\label{eq_safe2}
 & \forall {t}_0 \in \traces{\mytsys_{v'_0}}\ \exists {t}_1 \in \traces{\mytsys_{v'_1}}\cdots
 \forall {t}_{i-1} \in \traces{\mytsys_{v'_{i-1}}}\ \exists {t}_i \in \traces{\mytsys_{v'_i}} \notag\\
  & \quad \forall {t}_{i+1} \in \traces{\mytsys_{\Succ{v_{i+1}}}}\ \exists {t}_{i+2} \in \traces{\mytsys_{\Succ{v_{i+2}}}} \cdots \forall {t}_{k-2} \in \traces{\mytsys_{\Succ{v_{k-2}}}}\ \exists {t}_{k-1} \in \traces{\mytsys_{\Succ{v_{k-1}}}} : \notag\\ 
 & \quad\quad \combine{{t}_0,\ldots,{t}_{k-1}} \in L(\auti{k-1}[q_{k-1}]),
\end{align}
and
\begin{align}
\label{eq_proposition}
   \prophecies{v'_{i-1}} \text{ contains } \psafe{i-1}{q_{i+1}}{\bar{v}} \text{ with } \bar{v} = (v'_0,\ldots,v'_{i}).
\end{align}
Note that \cref{eq_proposition} for \round{r}[i] witnesses that \vplayer{i} used the \nonemptycase-case to pick her move in \round{r}[i], which was our original claim we prove by induction.

\subparagraph*{Inner Base Case, \boldmath\round{r}[1].}
In \round{r}[0], \cref{eq_safe} is true. We now restrict the considered "traces" for the first two quantifiers.
Clearly, the statement remains true when restricting the first quantifier which is universal to "traces" from "paths" starting with $v'_0 \in \Succ{v_0}$.
Furthermore, there exists some $v'_1 \in \Succ{v_1}$ such that the second quantifier which is existential can be restricted to all "traces" from "paths" starting with $v'_1$ and the statement remains true.
This is the statement of \cref{eq_safe2} for \round{r}[1].
Recall that 
\[
\safeprophy{0}{(q_{1})}{(v'_0,v'_1)} = \set{ t_0 \mid t_0 \in \traces{\mytsys_{v'_0}} \text{ and there is } t_1 \in \traces{\mytsys_{v'_1}} \text{ s.t.\ } \combine{t_0,t_1} \in L(\auti{1}[q_1])}.
\]
\cref{eq_safe2} for \round{r}[1] yields that 
\[
\safeprophy{0}{(q_{2})}{(v'_0,v'_1)} = \traces{\mytsys_{v'_0}}.
\]
Since "Falsifier" has moved to $v'_0$, he is committed to play one of those "paths", i.e., 
\[
\lambda^{\calP}(\rho_0)[r,\infty) \in \safeprophy{i-1}{(q_1)}{(v'_0,v'_1)}.
\]
\cref{assump_strategy} implies that "Falsifier" has truthfully indicated this, i.e., 
\[
\psafe{0}{(q_1)}{(v'_0,v'_1)} \in \prophecies{v'_0}.
\]
This is the statement of \cref{eq_proposition}.
Hence, \vplayer{1} can move to $v'_1$ in \round{r}[1], i.e., we are in the \nonemptycase-case.
This completes the inner base case.

\subparagraph*{Inner Inductive step, \boldmath\round{r}[i] to \boldmath\round{r}[i+2].}
We use similar arguments as for the inner base case.
By induction hypothesis, \cref{eq_safe2} for \round{r}[i] is true.
We now restrict the considered "traces" for the $(i+1)$-th and  $(i+2)$-th quantifier.
The $(i+1)$-th quantifier is universal, so clearly the statement remains true if we restrict the "traces" for this quantifier to "traces" coming from "paths" starting with $v'_{i+1} \in \Succ{v_{i+1}}$.
The $i+2$-th quantifier is existential, hence, there exists some $v'_{i+2} \in \Succ{v_{i+2}}$ such that this quantifier can be restricted to "traces" from "paths" starting with $v'_{i+2}$ and the statement remains true.
This is exactly the statement of \cref{eq_safe2} for \round{r}[i+2].
Recall that the "prophecy" $\safeprophy{i+1}{q_{i+2}}{(v'_0,\ldots,v'_{i+2})}$ is defined as
\begin{equation*}
\begin{split}
    &\{\combine{t_0,t_2,\ldots,t_{i+1}} \mid t_0 \in \traces{\mytsys_{v'_0}}, t_2 \in \traces{\mytsys_{v'_2}},\ldots,t_{i+1} \in \traces{\mytsys_{v'_{i+1}}} \\ 
     & \qquad \text{and there exist } t_{1} \in \traces{\mytsys_{v'_1}}, t_{3} \in \traces{\mytsys_{v'_3}},\ldots, t_{i+2} \in \traces{\mytsys_{v'_{i+2}}} \\
     & \qquad \qquad \text{such that } \combine{t_0,\ldots,t_{i+2}} \in L(\auti{i+2}[q_{i+2}]) \}. 
\end{split}
\end{equation*}
We claim (shown below) that \cref{eq_safe2} for \round{r}[i+2] implies that
\[
\safeprophy{i+1}{q_{i+2}}{(v'_0,\ldots,v'_{i+2})} = \traces{\mytsys_{v'_0}} \times \traces{\mytsys_{v'_2}} \times \cdots \times \traces{\mytsys_{v'_{i+1}}}.
\]
Consequently, 
\[
\combine{\lambda^{\calP}(\rho_0)[r,\infty),\lambda^{\calP}(\rho_2)[r,\infty),\ldots,\lambda^{\calP}(\rho_{i+1})[r,\infty)} \in \safeprophy{i+1}{q_{i+2}}{(v'_0,\ldots,v'_{i+2})}.
\]
In words, "Falsifier" has committed to play such that the remainders of his induced "traces" satisfy the above "prophecy" (which includes information about possible plays for "Verifier-players").
Due to \cref{assump_strategy}, "Falsifier" truthfully indicates this, i.e., 
\[
\psafe{i+1}{(q_1,q_3,\ldots,q_{i+2})}{(v'_0,\ldots,v'_{i+2})} \in \prophecies{v'_{i-1}}.
\]
This is the statement of \cref{eq_proposition}.
Hence, \vplayer{i+2} can move to $v'_{i+2}$ in \round{r}[i+2] which meets the requirements imposed by the strategy construction.

It is left to prove that 
\[
\safeprophy{i+1}{q_{i+2}}{(v'_0,\ldots,v'_{i+2})} = \traces{\mytsys_{v'_0}} \times \traces{\mytsys_{v'_2}} \times \cdots \times \traces{\mytsys_{v'_{i+1}}}.
\]
Towards a contradiction, assume there are $t_0 \in \traces{\mytsys_{v'_0}},t_2\in \traces{\mytsys_{v'_2}},\ldots,t_{i+1}\in \traces{\mytsys_{v'_{i+1}}}$ such that 
\[
\combine{t_0,t_2,\ldots,t_{i+1}} \notin \safeprophy{i+1}{q_{i+2}}{(v'_0,\ldots,v'_{i+2})}.
\]\
\cref{eq_safe2} for \round{r}[\mathmbox{i+2}], by successively restricting the quantifiers, implies that there are $t_1 \in  \traces{\mytsys_{v'_1}},t_3 \in  \traces{\mytsys_{v'_1}},\ldots,t_{i+2} \in  \traces{\mytsys_{v'_{i+2}}}$ such that $\combine{t_0,\ldots,t_{i+2}} \in L(\auti{i+2}[q_{i+2}])$.
This implies that 
\[
\combine{t_0,t_2,\ldots,t_{i+1}} \in \safeprophy{i+1}{q_{i+2}}{(v'_0,\ldots,v'_{i+2})}
\]
which contradicts the assumption.

Now that we have completed the inner induction, we return to the proof of \cref{eq_safe} for \round{r+1}.
We proved that \cref{eq_safe2} is true in \round{r}[k-1], meaning that 
\begin{equation*}
\begin{split}
 & \forall {t}_0 \in \traces{\mytsys_{v'_0}}\  \exists {t}_1 \in \traces{\mytsys_{v'_1}} \cdots
 \forall {t}_{k-2} \in \traces{\mytsys_{v'_{k-2}}}\  \exists {t}_{k-1} \in \traces{\mytsys_{v'_{k-1}}} : \\
 & \qquad \combine{{t}_0,\ldots,{t}_{k-1}} \in L(\auti{k-1}[q_{k-1}])
 \end{split}
\end{equation*}
Let $q'_{k-1} = \delta(q_{k-1},(\orig{{v'_0}},\ldots,\orig{{v'_{k-1}}}))$. We reformulate the previous equation as
\begin{equation}
\label{eq_step}
\begin{split}
 & \forall {t}_0 \in \traces{\mytsys_{v'_0}}\  \exists {t}_1 \in \traces{\mytsys_{v'_1}} \cdots
 \forall {t}_{k-2} \in \traces{\mytsys_{v'_{k-2}}}\  \exists {t}_{k-1} \in \traces{\mytsys_{v'_{k-1}}} : \\ 
 & \qquad \combine{{t}_0[1,\infty),\ldots,{t}_{k-1}[1,\infty)} \in L(\auti{k-1}[q'_{k-1}])
 \end{split}
\end{equation}
In \round{r}[\mathmbox{k-1}], the play is in a position~$((v'_0,\ldots,v'_{k-2},v_{k-1}),q,k-1)$.
As described, \vplayer{k-1} picks a suitable vertex $v'_{k-1} \in \Succ{v_{k-1}}$ and the play moves to position $((v'_0,\ldots,v'_{k-1}),q',0)$ where $q' = \delta_\autb(q,(\lambda^{\calP}({v'_0}),\ldots,\lambda^{\calP}({v'_{k-1}})))$ and is now in \round{r+1}[0].
So \cref{eq_step} clearly implies that \cref{eq_safe} is true in \round{r+1}.
We completed the proof of the outer induction, i.e., in every play that is consistent with the strategies constructed above and that satisfies \cref{assump_strategy}, the "Verifier-players" always use the \nonemptycase-case.

Next, we show that the strategies are consistent with the available observations.
We note that the strategy we constructed for \vplayer{i} is solely based on the information visible to her, i.e., it is based only on the "paths" $\rho_0,\ldots,\rho_i$ and derived information such as information about $\auti{i}$ (processing "traces" $\orig{\rho_0},\ldots,\orig{\rho_i}$) and related "prophecy" information.

It is left to argue that the resulting collection of strategies for the "Verifier-players" is winning.
We review the winning condition: 
Recall that the "parity automaton"~$\autb$ accepts "trace assignments" (stemming from $\tsysmanip$) satisfying $\psimanip$ which is the trace quantifier-free part of $\phimanip$.
The "Verifier-players" win if $\autb$ accepts $\combine{\lambda^{\calP}(\rho_0),\ldots,\lambda^{\calP}(\rho_{k-1})}$.

As mentioned at the beginning of the proof, \cref{assump_strategy} holds.
Recall that $\pi_0,\ldots,\pi_{k-1}$ are the quantified "trace variables" of $\phimanip$.
Under \cref{assump_strategy}, considering the "trace assignment" which maps each $\pi_j$ to $\lambda^{\calP}(\rho_j)$, the premise of $\psimanip$ is true.
Note that "traces" induced by "Verifier-players" do not influence the truth-value of the premise, because the "prophecies" only make predictions about "traces" induced by "Falsifier".
Hence, we have to show that the conclusion is true, i.e., that the above "trace assignment" satisfies $\mypsi$.
Recall that $\auti{k-1}$ checks this.
We have to show that $\auti{k-1}$ accepts $\combine{\orig{\rho_0},\ldots,\orig{\rho_{k-1}}}$.
Playing accordingly to the strategy we constructed, in each \round{r}[k-1], the play is in a position $((v'_0,\ldots,v'_{k-2},v_{k-1}),q,k-1)$ and \vplayer{k-1} has picked some $v'_{k-1} \in \Succ{v_{k-1}}$ such that 
\[
 \psafe{k-2}{q_{k-1}}{(v'_0,\ldots,v'_{k-1})} \in \prophecies{v'_{k-2}},
\]
where $q_{k-1}$ is the state of $\auti{k-1}$ reached on 
\[
\combine{\orig{\rho_0}[0,r),\orig{\rho_1}[0,r),\ldots,\orig{\rho_{k-1}}[0,r)}.
\]
If $L(\auti{k-1}[q_{k-1}]) = \emptyset$, it would not have been possible to indicate that $\safeprophy{k-2}{q_k}{(v'_0,\ldots,v'_{k-1})}$ holds.
Consequently, the state $q_{k-1}$ must be safe, because \auti{k-1} is a "safety automaton" since $\mypsi$ is a "safety" condition (safe simply means able to still accept something).
Thus, the run of $\auti{k-1}$ on $\combine{\orig{\rho_0},\ldots,\orig{\rho_{k-1}}}$ is accepting which implies that the run of $\autb$ on  $\combine{\lambda^{\calP}(\rho_0),\ldots,\lambda^{\calP}(\rho_{k-1})}$ is accepting.
\end{proof}

Combining \cref{thm_soundprophy} and \cref{thm_safetycompleteness}, we obtain our main result of this section for formulas~$\myphi$ where the maximal trace quantifier-free subformula~$\mypsi$ is a "safety" property.

\begin{theorem}
Let $\Xi$, $\calP$ be as in \cref{def_safetymanip}, and let $\mypsi$ be a "safety" property. The coalition of "Verifier-players" has a winning collection of strategies for $\game{\tsysmanip, \phimanip}$  if and only if $\mytsys \models \myphi$.
\end{theorem}

\section{\texorpdfstring{Complete Prophecies for {\setBold[0.5]$\omega$}-Regular Properties}{Complete Prophecies for Omega-Regular Properties}}
\label{sec_rabin}

In this section, we show how to construct complete prophecies for trace quantifier-free $\psi$ beyond the "safety" fragment.
In the "safety" case, "prophecies" indicate moves which prevent the "Verifier-players" from losing. 
In the case of "safety", not losing equals winning.
However for general $\omega$-regular objectives, this is not enough. It is also necessary to (again and again) make progress towards satisfying the acceptance condition.

In this section, we assume that the automata~$\auti{i}$ are \intro{Rabin automata}.
The class of languages recognized by (deterministic respectively nondeterministic) "Rabin automata" is the class of $\omega$-regular languages.
A (deterministic) "Rabin automaton" is a (deterministic) $\omega$-automaton whose acceptance condition is given as a set~$\set{(B_1,F_1),\ldots,(B_m,F_m)}$ of pairs of sets~$B_i, F_i$ of states. 
A run is accepting if there exists an $i$ such that the run visits states in $B_i$ only finitely many times and states in $F_i$ infinitely many times.

\subsection{One-Pair Rabin Automata}
\label{subsec_onepairrabin}
First, we present our construction for one-pair "Rabin automata", then show how to generalize to "Rabin automata" with an arbitrary number of pairs in \cref{subsec_rabin}.
Here, we assume that for \emph{each} $i \in \iodd$, the automaton $\auti{i}$ is a one-pair "Rabin automaton", not only $\auti{k}$.
We always denote the single pair as $(B,F)$, it will be clear from the context to which automaton it belongs.
Intuitively, making progress towards winning in a one-pair "Rabin automaton" means avoiding states in $B$ and visiting states in $F$.

Let $\kappa$ be an infinite run of such an automaton.
\AP We define $\intro*{\lastb{\kappa}}$ as $\bot$ if $\kappa$ never visits~$B$, as $i$ if $\kappa(i)$ is the last visit of $\kappa$ in $B$, and as $\infty$ if $\kappa$ visits $B$ infinitely many times (meaning there is no last visit).
\AP Furthermore, we define $\intro*{\firstf{\kappa}}$ as $i$ if $\kappa(i)$ is the first visit of $\kappa$ in $F$ and as $\infty$ if $\kappa$ never visits $F$.
\AP Finally, we define the value $\intro*{\val{\kappa}}$ of a run $\kappa$ as $\firstf{\kappa}$ if $\lastb{\kappa} = \bot$ and as $\max\set{\lastb{\kappa},\firstf{\kappa}}$ otherwise.
Given two runs $\kappa,\kappa'$, we say that~$\kappa$ \myquot{makes progress faster} than $\kappa'$ if $\val{\kappa} < \val{\kappa'}$.

Before we can define our new "prophecies", we need additional notation that gives us the value of a best run (over several "traces") if we fix all but the last "trace".
\AP For each $i \in \iodd$, each vertex $v$ of~$\mytsys$, each state $q$ of $\auti{i}$, and traces $t_0,t_1,\ldots,t_{i-1} \in \traces{\mytsys}$, we define
\[
\intro*{\opt{i}{q,v,t_0,\ldots,t_{i-1}}} = \min\nolimits_{t_i \in \traces{\mytsys_{\Succ{v}}}} \val{\kappa_{i}},
\]
where $\kappa_{i}$ is the run of $\auti{i}[q]$ on $\combine{t_0,\ldots,t_{i-1},t_i}$.

We are now ready to introduce our "prophecies".

\begin{definition}[Prophecy construction]
\AP 
For each $i \in \ieven$, vectors $\bar{u} = (v_0,v_1,\ldots,v_{i+1})$ and\newline $\bar{v} = (v'_0,v'_1,\ldots,v'_{i+1})$ of vertices such that $v'_{j} \in \Succ{v_{j}}$ for $j \leq i+1$, and each vector $\bar{q} = (q_1,q_3,\ldots,q_{i+1})$ of states~$q_j$ of $\auti{j}$, we define
\begin{equation*}
\begin{split}
 & \intro*{\progressprophy{i}{\bar{q}}{\bar{u},\bar{v}}} = 
 \{\combine{t_0,t_2,\ldots,t_{i}} \mid t_0 \in \traces{\mytsys_{v'_0}}, t_2 \in \traces{\mytsys_{v'_2}},\ldots,t_{i} \in \traces{\mytsys_{v'_{i}}}, \text { and} \\ 
 & \quad \text{there exist } t_{1} \in \traces{\mytsys_{v'_1}}, t_{3} \in \traces{\mytsys_{v'_3}},\ldots, t_{i+1} \in \traces{\mytsys_{v'_{i+1}}} \text{ such that } \\
 & \quad \quad \text{the run } \kappa_j \text{ of } \auti{j}[q_{j}] \text{ on } \combine{t_0,\ldots,t_{j}} \text{ is accepting and satisfies}\\
 & \quad\quad\quad \val{\kappa_j} = \opt{j}{q_{j},v_j,t_0,\ldots,t_{j-1}} \text{ for all odd }j \leq i+1\}. 
\end{split}
\end{equation*}
\end{definition}

Intuitively, the "prophecy" \progressprophy{i}{\bar{q}}{\bar{u},\bar{v}} indicates that it is safe for \vplayer{i+1} to move from $v_{i+1}$ to $v'_{i+1}$ as before, and additionally, among all successors of $v_{i+1}$, picking $v'_{i+1}$ allows \vplayer{i+1} to make the most progress.
This "prophecy" is intended to be used, if for odd $j < i+1$, \vplayer{j} already picked the move $v_j$ to $v'_j$ which was optimal.
In particular, the "prophecy" repeats the "prophecies" made (in the same round) for \vplayer{j} with odd $j < i+1$.

These "prophecies" are $\omega$-regular.

\begin{lemma}[name={},restate=Prophregular]
\label{proph_regular}
For each $i \in \ieven$, each vector~$\bar{q}$ of states and all vectors~$\bar{u},\bar{v}$ such that the $\ell$-th vertex of $\bar{v}$ is a successor of the $\ell$-th vertex of $\bar{u}$, the set \progressprophy{i}{\bar{q}}{\bar{u},\bar{v}} is $\omega$-regular.  
\end{lemma}

\begin{proof}
For ease of presentation, we speak of value of a "trace" $t$ when we actually mean the value of the unique run $\kappa$ on $t$.
We also sometimes write $\val{t}$ instead of $\val{\kappa}$. 
In both cases, the automaton we consider will be clear from context.

Let $\mathcal{R}$ be a one-pair "Rabin automaton".
One can construct an $\omega$-automaton that recognizes 
\[
\set{\combine{t,t'} \mid \val{\kappa} > \val{\kappa'}, \text{ where } \kappa \text{ resp.\ } \kappa' \text{ is the run of } \mathcal{R} \text{ on } t \text{ resp.\ }t'}
.\]
This automaton checks whether one of the following cases holds:
\begin{enumerate}
    \item $\kappa$ visits $B$ infinitely many times (hence, $\val{\kappa} = \infty$), $\kappa'$ visits $B$ finitely many times and $F$ at least once (hence, $\val{\kappa'} < \infty$).
    \item $\kappa$ and $\kappa'$ visit $B$ finitely many times and $F$ at least once, and
    \begin{enumerate}
        \item $\lastb{\kappa},\lastb{\kappa'},\firstf{\kappa'} < \firstf{\kappa}$, or \item $\firstf{\kappa},\firstf{\kappa'},\lastb{\kappa'} < \lastb{\kappa}$
    \end{enumerate}
    \item $\kappa$ and $\kappa'$ do not visit $B$, and $\firstf{\kappa'} < \firstf{\kappa}$. 
\end{enumerate}

By projection to the first component of the set, one obtains an automaton that recognizes "traces" $t$ such that there is some "trace" $t'$ with $\val{t} > \val{t'}$.
The complement of this set (in intersection with valid "traces") yields "traces" that have an optimal value for $\mathcal R$.

Using these techniques, one can construct an automaton $\autb_1$ that accepts those $\combine{t_0,t_1}$ where $\combine{t_0,t_1} \in L(\auti{1})$ such that $t_1$ is optimal in combination with $t_0$ for $\auti{1}$.
Again, using these techniques, one can construct an automaton $\autb_3$ that accepts those $\combine{t_0,t_1,t_2,t_3}$ where $\combine{t_0,t_1} \in L(\autb_1)$, $\combine{t_0,t_1,t_2,t_3} \in L(\auti{3})$ and $t_3$ is optimal in combination with $t_0,t_1,t_2$ for $\auti{3}$.

Iterating this, one arrives at an automaton $\autb_{i}$ that accepts those $\combine{t_0,t_1,\ldots,t_i}$.
Projecting away the "traces" $t_j$ for odd $j \leq i$, one obtains an automaton that accepts those $\combine{t_0,t_2,\ldots,t_{i-1}}$ required by the "prophecy" $\progressprophy{i-1}{\bar{q}}{\bar{u},\bar{v}}$ (for some $\bar{q},\bar{u},\bar{v}$).
\end{proof}

We set up the manipulated system and formula, then state the completeness result.

\begin{definition}
\label{def_manip}
 We define  $\Xi = (\Xi_i)_{i \in \ieven}$ and $\calP = (\pp_i)_{i \in \ieven}$.
 Let $\Xi_{i}$ be the set of \qptl formulas that contains sentences $\pformula{i}{\bar{q}}{\bar{u},\bar{v}}$ for each vector~$\bar{q}$ of states and vectors $\bar{u},\bar{v}$ such that the $\ell$-th vertex of $\bar{v}$ is a successor of the $\ell$-th vertex of $\bar{u}$.
 The sentence $\pformula{i}{\bar{q}}{\bar{u},\bar{v}}$ expresses the "prophecy"~$\progressprophy{i}{\bar{q}}{\bar{u},\bar{v}}$ using only even "trace variables" $\pi_j$ with $j \leq i$.
 
 \AP Moreover, let~$\pp_{i}$ be such that it contains the "prophecy variable" $\intro*{\pprogress{i}{\bar{q}}{\bar{u},\bar{v}}}$ for each vector~$\bar{q}$ of states and vectors $\bar{u},\bar{v}$ such that the $\ell$-th vertex of $\bar{v}$ is a successor of the $\ell$-th vertex of $\bar{u}$.
 The "prophecy variable" $\pprogress{i}{\bar{q}}{\bar{u},\bar{v}}$ corresponds to $\pformula{i}{\bar{q}}{\bar{u},\bar{v}}$.
\end{definition}

\begin{lemma}[name={},restate=Completeness]
\label{thm_completness}
 Assume $\mytsys \models \myphi$ and each $\auti{i}$ is a one-pair "Rabin automaton".
 Then the coalition of "Verifier-players" has a winning collection of strategies in $\game{\tsysmanip, \phimanip}$ with $\Xi$ and~$\calP$ as in \cref{def_manip}.
\end{lemma}

Our proof of this result generalizes the one for the "safety" case (see \cref{thm_safetycompleteness}), using the same notation introduced there:
We define a collection of strategies for the "Verifier-players" and then show, that $\mytsys \models \myphi$ implies that this collection is winning in $\game{\tsysmanip, \phimanip}$.

Assume the play is in \round{r}[i] for $r \geq 0$ and $i \in \iodd$.
Thus, it is in a position
\[
((v'_0,\ldots,v'_{i-1},v_i,\ldots,v_{k-1}),q,i)
\]
and \vplayer{i} has to move.
We review some information \vplayer{i} can use to base her choice on.
In \round{r}, \vplayer{i} has access to the first $r$ letters of the "traces" $\orig{\rho_0},\orig{\rho_1},\ldots,\orig{\rho_i}$ induced by the play\textsuperscript{\ref{foot_strategy}}.
Let $q_{j}$ be the state that $\auti{j}$ has reached on $\combine{\orig{\rho_0}[0,r),\orig{\rho_1}[0,r),\ldots,\orig{\rho_j}[0,r)}$ for odd $j \leq i$.
Also, let $((v_0,\ldots,v_{i-1},v_i,\ldots,v_{k-1}),q,0)$ be the position reached in \round{r}[0].

\begin{definition}[Strategy construction]
\label{def_strategy}
\AP Let
\[\possmoves{i}{\bar{u},\bar{v},\bar{q}} = \set{ v'_i \in \Succ{v_i} \mid \pprogress{i-1}{\bar{q}}{\bar{u},\bar{w}} \in \prophecies{v'_{i-1}} }\]
where $\bar{u} = (v_0,\ldots,v_i)$, $\bar{v} = (v_0',\ldots,v'_{i-1},v_i)$, $\bar{w} = (v_0',\ldots,v'_{i-1},v'_i)$ and $\bar{q} = (q_1,q_3,\ldots,q_{i})$.

If $\possmoves{i}{\bar{u},\bar{v},\bar{q}}$ is nonempty, then \vplayer{i} can move to any $v'_i$ in the set (the \myquot{\nonemptycase-case}).
If it is empty, then \vplayer{i} can move to any $v'_i \in \Succ{v_i}$ (the \myquot{\emptycase-case}).
\end{definition}

Now, \cref{thm_completness} can be proven by generalizing the proof of \cref{thm_safetycompleteness}.

We begin by showing that if the "Verifier-players" use the strategies defined in \cref{def_strategy}, then the "Verifier-players" have always used the \nonemptycase-case.
This is under the assumption that "Falsifier" adheres to \cref{assump_strategy}, that is, he picks "paths" that start in the intended sub-parts of the manipulated "transition system", and furthermore, he always truthfully indicates which "prophecies" hold.

\begin{lemma}
\label{lemma_safetytoprogress}
    Let $\alpha$ be a play in $\game{\tsysmanip,\phimanip}$ that is consistent with the strategies constructed in \cref{def_strategy} and that satisfies \cref{assump_strategy}, then the "Verifier-players" have always used the \nonemptycase-case to pick their moves.
\end{lemma}

\begin{proof}
Let $\rho_0,\ldots,\rho_{k-1}$ be the "paths" induced by $\alpha$.
For each $r \geq 0$, and each $i\in\iodd$, we show inductively that \vplayer{i} can use the \nonemptycase-case in \round{r}[i] to pick her move. 
Our induction hypothesis is that in each previous round of a "Verifier-player" the \nonemptycase-case was used.

We assume that in \round{r}[0] the play was in a position of the form 
\[
((v_0,\ldots,v_{i-1},v_i,\ldots,v_{k-1}),q,0),
\]
and in \round{r}[i] the play is in a position of the form
\[
((v'_0,\ldots,v'_{i-1},v_i,\ldots,v_{k-1}),q,i).
\]
Furthermore, let $q_{j}$ be the state that $\auti{j}$ has reached on 
\[
\combine{\orig{\rho_0}[0,r),\orig{\rho_1}[0,r),\ldots,\orig{\rho_j}[0,r)}
\]
for odd $j \leq i$.

Let $\bar{u} = (v_0,\ldots,v_i)$, $\bar{v} = (v'_0,\ldots,v'_{i-1},v_i)$ and $\bar{q} = (q_1,q_3,\ldots,q_{i})$.
Our goal is to show that there is some $v'_i \in \Succ{v_i}$ such that for $\bar{w} = (v'_0,\ldots,v'_{i-1},v'_i)$ holds that 
\begin{equation}
\label{eq_nonempty}
\combine{\orig{\rho_0}[0,r),\orig{\rho_2}[0,r),\ldots,\orig{\rho_{i-1}}[0,r)} \in \progressprophy{i-1}{\bar{q}}{\bar{u},\bar{w}}
\end{equation}
which implies that $\pprogress{i-1}{\bar{q}}{\bar{u},\bar{w}} \in \prophecies{v'_{i-1}}$ (because "Falsifier" adheres to \cref{assump_strategy}).
By definition, then $v'_i \in \possmoves{i}{\bar{u},\bar{v},\bar{q}}$, and \vplayer{i} is in the \nonemptycase-case.

We first make an observation about the connection between "prophecies" used in the "safety" case and our new "prophecies":

For each $j \in \iodd$, vertex vectors $\bar{x},\bar{y}$ such that the $\ell$-th vertex of $\bar{y}$ is a successor vertex of the $\ell$-th vertex of $\bar{x}$, each state vector $\bar{s}$, if
\[
\combine{t_0,t_2,\ldots,t_{j-1}} \in \progressprophy{j-1}{\bar{s}}{\bar{x},\bar{y}},
\]
then
\[
\combine{t_0,t_2,\ldots,t_{j-1}} \in \safeprophy{j-1}{s}{y},
\]
where $s$ is the last state of $\bar{s}$ and $y$ is the last vertex of $\bar{y}$.
Note that the definition of the "safety" "prophecies"~$\safeprophy{j-1}{s}{y}$ does \emph{not} depend on the automaton inducing the "prophecies" being a "safety automaton":
They just ensure that the remaining suffixes induced by the play can be accepted, if the "Verifier-players" play accordingly. 
Similarly, the induction at the beginning of the proof of \cref{thm_safetycompleteness}, showing that the "Verifier-players" can always use the \nonemptycase-case, is also independent of the acceptance condition.

As a consequence, this allows us to reuse (without reproving) \cref{eq_safe} and \cref{eq_safe2}.
To recap, we want to prove \cref{eq_nonempty} for \round{r}[i].
Therefore, we first prove that \cref{eq_safe} for \round{r}[0] implies \cref{eq_nonempty} for \round{r}[1].
Then, we prove that \cref{eq_safe2} for \round{r}[i] together with \cref{eq_nonempty} for \round{r}[i-2] implies \cref{eq_nonempty} for \round{r}[i] with $i > 1$.

We show \cref{eq_nonempty} for \round{r}[1].
We have to show that $\orig{\rho_0}[0,r) \in \progressprophy{0}{(q_1)}{(v_0,v_1),(v'_0,v'_1)}$ for some $v'_1 \in \Succ{v_1}$.
Let $t_0$ denote $\orig{\rho_0}[r,\infty) \in \traces{\mytsys_{v'_0}}$.
\cref{eq_safe} implies that there is some $t_1 \in \traces{\mytsys_{\Succ{v_1}}}$ such that $\auti{1}[q_1]$ accepts $\combine{t_0,t_1}$.
Hence, there is a witness $t'_1 \in \traces{\mytsys_{\Succ{v_1}}}$ such that the run $\kappa_1$ of $\auti{1}[q_1]$ on $\combine{t_0,t'_1}$ is accepting and $\val{\kappa_1} = \opt{1}{q_1,v_1,t_0}$.
Clearly, $t'_1$ in $\traces{\mytsys_{v'_1}}$ for some $v'_1 \in \Succ{v_1}$.
Thus, $t_0 \in \progressprophy{0}{(q_1)}{(v_0,v_1),(v'_0,v'_1)}$.
We completed the proof for \round{r}[1].

We show \cref{eq_nonempty} for \round{r}[i] with $i > 1$.
Let $t_j$ denote $\orig{\rho_j}[r,\infty) \in \traces{\mytsys_{v'_j}}$ for all even $j < i$.
By assumption, \cref{eq_nonempty} holds in \round{r}[i-2].
Consequently, 
\[
\combine{t_0,t_2,\ldots,t_{i-3}} \in \progressprophy{i-3}{(q_1,q_3,\ldots,q_{i-2})}{(v_0,\ldots,v_{i-2}),(v'_0,\ldots,v'_{i-2})}.
\]
By definition of the above "prophecy", we obtain that there is a witness $t'_j \in \traces{\mytsys_{v'_j}}$ such that 
\[
\val{\kappa_j} = \opt{j}{q_{j+1},v_j,t_0,t'_1,t_2,t'_3,\ldots,t_{j-1}}, \text{ and $\kappa_j$ is accepting},
\]
where $\kappa_j$ is the run of $\auti{j}[q_{j}]$ on $\combine{t_0,t'_1,t_2,t'_3,\ldots,t_{j-1},t'_j}$ for all odd $j < i$.
Note that the "traces" with odd indices are the witnesses of optimality and the "traces" with even indices are the traces of "Falsifier".

\cref{eq_safe2} for \round{r}[i] implies that there is some $t \in \traces{\mytsys_{\Succ{v_i}}}$ such that $\auti{i}[q_{i}]$ accepts \[\combine{t_0,t'_1,t_2,t'_3,\ldots,t_{i-1},t}.\]
Hence, there is a witness $t'_i \in \traces{\mytsys_{\Succ{v_i}}}$ such that \[
\val{\kappa_i} = \opt{i}{q_{i+1},v_i,t_0,t'_1,t_2,t'_3,\ldots,t_{i-1}}, \text{ and $\kappa_i$ is accepting},
\]
where $\kappa_i$ is the run of $\auti{i}[q_{i}]$ on $\combine{t_0,t'_1,t_2,t'_3,\ldots,t_{i-1},t'_i}$.

Clearly, $t'_i \in \traces{\mytsys_{v'_i}}$ for some $v'_i \in \Succ{v_i}$.
Thus, $\combine{t_0,t_2,\ldots,t_{i-1}} \in \progressprophy{i-1}{\bar{q}}{\bar{u},\bar{w}}$ where the last vertex of $\bar{w}$ is $v'_i$.
We have completed the proof for \round{r}[i].
\end{proof}


We show an auxiliary lemma (\cref{lemma_progress}) that allows us to prove that playing according to the proposed strategies ensures that the play induces "traces" that are accepted by the $\auti{i}$ (\cref{lemma_accept}).
Intuitively, \cref{lemma_progress} states that making one move according to the proposed strategy definition (\cref{def_strategy}) yields progress towards accepting.

\begin{lemma}
\label{lemma_progress}
For each $i \in \iodd$, vectors $\bar{u} = (v_0,v_1,\ldots,v_{i})$ and $\bar{v} = (v'_0,v'_1,\ldots,v'_{i})$ of vertices such that $v'_{j} \in \Succ{v_{j}}$ for $j \leq i$, each vector $\bar{q} = (q_1,q_3,\ldots,q_{i})$ of states, and "traces" $t_j \in \traces{\mytsys_{v'_j}}$ for $j < i$ such that 
\[
\val{\kappa_j} = \opt{j}{q_{j},v_j,t_0,\ldots,t_{j-1}}, \text{ and $\kappa_j$ is accepting},
\]
where $\kappa_j$ is the run of $\auti{j}$ on $\combine{t_0,\ldots,t_j}$ for all odd $j < i$, and 
\[\combine{t_0,t_2,\ldots,t_{i-1}} \in \progressprophy{i-1}{\bar{q}}{\bar{u},\bar{v}}\]
holds the following:
\[
\opt{i}{q_{i},v_{i},t_0,\ldots,t_{i-1}} < \infty,
\]
and if $q_{i} \notin F$ then
\[
\opt{i}{q'_{i},v'_{i},t_0[1,\infty),\ldots,t_{i-1}[1,\infty)} < \opt{i}{q_{i},v_{i},t_0,\ldots,t_{i-1}},
\]
where $q'_{i} = \delta_{i}(q_{i},\lambda(v'_0),\ldots,\lambda(v'_{i}))$ and $\delta_{i}$ is the transition function of $\auti{i}$.    
\end{lemma}

\begin{proof}
Since $\val{\kappa_j} = \opt{j}{q_{j},v_j,t_0,\ldots,t_{j-1}}$ and $\kappa_j$ is accepting for all odd $j < i$, and \[\combine{t_0,t_2,\ldots,t_{i-1}} \in \progressprophy{i-1}{\bar{q}}{\bar{u},\bar{v}},\] we obtain, using the definition of $\progressprophy{i-1}{\bar{q}}{\bar{u},\bar{v}}$, that
\[
\opt{i}{q_{i},v_{i},t_0,\ldots,t_{i-1}} < \infty.
\]
Furthermore, the above implies also that there is a "trace" $t_i \in \traces{\mytsys_{v'_i}}$ such that 
\begin{equation}\label{eq_witness}
    \opt{i}{q_{i},v_{i},t_0,\ldots,t_{i-1}} = \val{\kappa_i},
\end{equation}
where $\kappa_i$ is the run of $\auti{i}[q_{i}]$ on $\combine{t_0,\ldots,t_{i-1},t_i}$.

We show that
\begin{equation}
\label{eq_optlessthanval}
\opt{i}{q'_{i},v'_{i},t_0[1,\infty),\ldots,t_{i-1}[1,\infty)} \leq \val{\kappa_i[1,\infty)}.
\end{equation}
For \cref{eq_optlessthanval} to hold, it suffices that $t(0) = \lambda(v'_i)$ and that $\kappa_i[1,\infty)$ is a run of $\auti{i}[q'_{i}]$.
Both conditions are true.

Furthermore, we show that
\begin{equation}
\label{eq_valstep}
\val{\kappa_i[1,\infty)} + 1 = \val{\kappa_i}.
\end{equation}
Since $\val{\kappa_i} < \infty$, we have $\lastb{\kappa_i} \neq \infty$ and $\firstf{\kappa_i} < \infty$.
If $\lastb{\kappa_i} = \bot$, then $\lastb{\kappa_i[1,\infty)} = \bot$, otherwise $\lastb{\kappa_i[1,\infty)} + 1 = \lastb{\kappa_i}$.
Furthermore, $\firstf{\kappa_i[1,\infty)} + 1 = \firstf{\kappa_i}$, because $q_{i+1} \neq F$ and $\kappa_i$ starts in $q_{i+1}$.
\cref{eq_valstep} follows.

Finally, by combining \cref{eq_witness,eq_optlessthanval,eq_valstep}, we obtain
\begin{align*}
    \opt{i}{q_{i},v_{i},t_0,\ldots,t_{i-1}} & = \val{\rho} > \val{\rho[1,\infty)} \geq \opt{i}{q'_{i},v'_{i},t_0[1,\infty),\ldots,t_{i-1}[1,\infty)}.\qedhere
\end{align*}
\end{proof}

We show that playing according to the proposed strategies yields "traces" that make optimal progress towards being accepted (for each of the $\auti{i}$) implying that in the limit they are actually accepted.
Ultimately, to show that the "Verifier-players" win a play in the game~$\game{\tsysmanip,\psimanip}$, it suffices to show that $\auti{k-1}$ accepts the "traces" induced by the play.
However, to prove this, we need that all $\auti{i}$ for $i \in \iodd$ accept their input "traces", to ensure that the "traces" up to $k-2$ yield a valid basis to be completed with the $(k-1)$-th "trace". 

\begin{lemma}
\label{lemma_accept}
Let $\alpha$ be a play in $\game{\tsysmanip,\phimanip}$ that is consistent with the strategies constructed in \cref{def_strategy} and that satisfies \cref{assump_strategy}.
Let $\rho_0^\alpha,\ldots,\rho_{k-1}^\alpha$ be the "paths" induced by $\alpha$, and let $t_i = \orig{\rho_i^\alpha}$ for $i \in \set{0,\ldots,k-1}$.
For $r \geq 0$ and $i\in\iodd$, we obtain  
\begin{equation}
\label{eq_opti}
   \val{\kappa_i}[r,\infty) = \opt{i}{\kappa_i(r),\rho_i^\alpha(r),t_0[r,\infty),\ldots,t_{i-1}[r,\infty)},  
\end{equation}
where $\kappa_i$ is the run of $\auti{i}$ on $\combine{t_0,\ldots,t_i}$.
Furthermore, $\combine{t_0,\ldots,t_i} \in L(\auti{i})$.
\end{lemma}

\begin{proof}
We proceed by induction on $i \in \iodd$.
We prove the statement for $i$ and assume the statement holds for odd $j < i$.

Recall that in \round{r}[i] with $r \geq 0$, we are in a position of the form
\[
((v'_0,\ldots,v'_{i-1},v_i,\ldots,v_{k-1}),q,i),
\]
where $\rho_i^\alpha(r) = v_i$.
We define
\[
x_r := \opt{i}{\kappa_i(r),v_i,t_0[r,\infty),\ldots,t_{i-1}[r,\infty)}.
\]
We claim that for any $r \geq 0$, we have $x_r < \infty$ and either $\kappa_i(r) \in F$ or $x_r > x_{r+1}$.
This implies that $\kappa_i$ is an accepting run:
The run $\kappa_i$ only visits finitely many times $B$, because $x_0 < \infty$.
And $\kappa_i$ visits $F$ infinitely many times, because there is no infinitely long decreasing sequence of natural numbers.

\cref{assump_strategy} yields that
\begin{equation}
\label{eq_prophy}
\combine{t_0[r,\infty),t_2[r,\infty),\ldots,t_{i-1}[r,\infty])} \in \progressprophy{i-1}{\bar{q}}{\bar{u},\bar{v}}, 
\end{equation}
where $\bar{u} = (v_0,\ldots,v_i)$, $\bar{v} = (v'_0,\ldots,v'_i)$, and $\bar{q} =(\kappa_1(r),\kappa_3(r),\ldots,\kappa_{i}(r))$, and $v'_\ell \in \Succ{v_\ell}$ is the move chosen in \round{r}[\ell], i.e., $v_{\ell} = \rho_{\ell}^\alpha(r)$ and $v'_{\ell} = \rho_{\ell}^\alpha(r+1)$ for $\ell \leq i$.

By induction hypothesis, the statement of \cref{lemma_accept} is true for all odd $j < i$.
Note that $\combine{t_0,\ldots,t_j} \in \auti{j}$ implies that $\combine{t_0[r,\infty),\ldots,t_j[r,\infty)} \in \auti{j}[\kappa_j(r)]$.
Thus, using \cref{eq_prophy} and \cref{eq_opti} for all odd $j < i$, the pre-conditions of \cref{lemma_progress} are satisfied.
Consequently, we get $x_r < \infty$, and if $\kappa_i(r) \notin F$, 
\begin{align*}
    x_{r+1} & = \opt{i}{\kappa_i(r+1),v'_i,t_0[r+1,\infty),\ldots,t_{i-1}[r+1,\infty)} \\
    & < \opt{i}{\kappa_i(r),v_i,t_0[r,\infty),\ldots,t_{i-1}[r,\infty)} = x_r,
\end{align*}
which shows that $\kappa_i$ is accepting, i.e., $\combine{t_0,\ldots,t_i} \in L(\auti{i})$.

It is left to prove that \cref{eq_opti} holds, i.e., for all $r\geq 0$, we have 
\[
x_r = \opt{i}{\kappa_i(r),\rho_i^\alpha(r),t_0[r,\infty),\ldots,t_{i-1}[r,\infty)} = \val{\kappa_i}[r,\infty).
\]
Recall that $\rho_i^\alpha(r) = v_i$ and that $\rho_i^\alpha(r+1) = v'_i$, and let $\kappa_i(r) = q$.
By definition of $\opt{i}{\cdot}$, there is a "path" $\rho$ with $\rho(0) \in \Succ{v_i}$ such that $\val{\kappa} = x_r$, where $\kappa$ is the run of $\auti{i}[q]$ on $\combine{t_0[r,\infty),\ldots,t_{i-1}[r,\infty),\lambda(\rho)}$.
By \cref{eq_prophy} and \cref{eq_opti} for all odd $j < i$, we obtain that $\rho$ can be chosen such that $\rho(0) = v'_i \in \Succ{v_i}$.

Hence, the "path" $\rho_i^\alpha$ ensures that $x_r = \val{\kappa_i}[r,\infty)$ is true for all $r \geq 0$.
\end{proof}


Finally, we are ready to prove completeness.
\begin{proof}[Proof of \cref{thm_completness}.]
We prove that the strategies defined above form a winning collection.
To this end, let $\alpha$ be a play in $\game{\tsysmanip, \phimanip}$ that is consistent with the strategies and satisfies \cref{assump_strategy}.
Plays in which the assumption is violated are trivially won by the "Verifier-players".
Furthermore, let $\rho_0^\alpha,\ldots,\rho_{k-1}^\alpha$ be the paths induced by $\alpha$, and let $t_i = \orig{\rho_i^\alpha}$ for $i \in \set{0,\ldots,k-1}$.
Since the premise of $\psimanip$ holds, we must show that its conclusion $\mypsi$ holds, i.e., that $\auti{k-1}$ accepts  $\combine{t_0,\ldots,t_{k-1}}$.
This is implied by \cref{lemma_accept} for $i = k-1$.
\end{proof}

We directly obtain our main result by combining \cref{thm_soundprophy} and \cref{thm_completness}.

\begin{theorem}
\label{thm_main}
Let $\Xi$ and $\calP$ as in \cref{def_manip}, and assume that each $\auti{i}$ is a one-pair "Rabin automaton".
The coalition of "Verifier-players" has a winning collection of strategies for $\game{\tsysmanip, \phimanip}$  if and only if $\mytsys \models \myphi$.
\end{theorem}

\subsection{Beyond One-Pair Rabin Automata}
\label{subsec_rabin}

We have proven \cref{thm_main} for the case that all automata referred to in the "prophecies" are one-pair "Rabin automata".
We sketch how to drop this constraint closely following Beutner and Finkbeiner~\cite{BF}.

A "Rabin automaton" $\mathcal R$ with pairs $(B_1,F_1),\ldots,(B_m,F_m)$ can be seen as a union of $m$ one-pair "Rabin automata"~$(B_i,F_i)$-$\mathcal R$. 
A run that is accepting in $(B_i,F_i)$-$\mathcal R$ is also accepting in $\mathcal R$ and each accepting run of $\mathcal{R}$ is accepting in some $(B_i,F_i)$-$\mathcal R$.
Before we present how to go from one-pair to general "Rabin automata", we recall that the "prophecies" exclusively talk about "traces" provided by "Falsifier". 
Hence, if "Falsifier" signals that his future "traces" belong to "prophecy" $\calP$ which, for example, implies that "Rabin automaton"~$\mathcal{R}$ accepts them, then "Falsifier" can also specify the pair(s) for which $\mathcal{R}$ will accept. No further context from the "Verifier-players" is needed.

The idea is to annotate our "prophecies" (which refer to $\auti{1},\auti{3},\ldots,\auti{k-1}$) with the indices of the pairs used for acceptance.
More concretely, for a "prophecy" $\progressprophy{i}{\bar{q}}{\bar{u},\bar{v}}$, we introduce "prophecies" of the form $\progressprophy{i}{\bar{q},\bar{a}}{\bar{u},\bar{v}}$, where $\bar{a}$ is a vector of indices of the same length as $\bar{q}$.
Say the $j$-th entry of $\bar{q}$ is $q_\ell$, then the "prophecy" $\progressprophy{i}{\bar{q}}{\bar{u},\bar{v}}$ refers to $\auti{\ell}[q_\ell]$ (among other automata), and say the $j$-th entry of $\bar{a}$ is $a_j$, then $\progressprophy{i}{\bar{q},\bar{a}}{\bar{u},\bar{v}}$ refers to the one-pair "Rabin automaton"~$(B_{a_j},F_{a_j})$-$\auti{\ell}[q_\ell]$ (among other automata).
Except for the refined "prophecies", no other changes to the game are required.

In a play of the game $\game{\tsysmanip,\phimanip}$, "Falsifier" specifies with his "prophecies" for which pairs the "Rabin automata"~$\auti{1},\auti{3},\ldots,\auti{k-1}$ can accept.

In \round{0}[0], it is specified for $\auti{1}$, in \round{0}[1], \vplayer{1} chooses one option specified by "Falsifier" for $\auti{1}$.
That choice is simply made by her move.
In \round{0}[0], "Falsifier" moves from $\bullet$ to $v_0$, and in \round{0}[1], \vplayer{1} moves from $\bullet$ to $v_1$.
Say $\prophecies{v_0}$ indicates that $(B_i,F_i)$-$\auti{1}$ accepts when \vplayer{1} moves to $v_1$, then she has chosen $(B_i,F_i)$-$\auti{1}$.
She then (in order be guaranteed to build a winning play) has to stick to it for the rest of the play, meaning that from now on all her moves must have a target vertex~$v$ such that the "prophecy" given by "Falsifier" in the move before indicates that $(B_i,F_i)$-$\auti{1}$ accepts if she moves to $v$.

In \round{0}[2], it is specified also for $\auti{3}$, in \round{0}[3], \vplayer{3} chooses one option for $\auti{3}$ and then has to stick to it for the rest of the play.
Additionally, she must pick the option chosen by \vplayer{1} for $\auti{1}$ to ensure consistency.
In \round{0}[2], "Falsifier" moves from $\bullet$ to $v_2$. 
For example, say \vplayer{3} has the option to move from $\bullet$ to $v'$, $v''$, or $v'''$ with $\prophecies{v_2}$ indicating that $(B_i,F_i)$-$\auti{1}$ and $(B_j,F_j)$-$\auti{3}$ accept if she moves to $v'$, that $(B_k,F_k)$-$\auti{1}$ and $(B_j,F_j)$-$\auti{3}$ accept if she moves to $v''$, and that $(B_i,F_i)$-$\auti{1}$ and $(B_\ell,F_\ell)$-$\auti{3}$ accept if she moves to $v'''$.
Then she must pick $v'$ or $v'''$ because \vplayer{1} has committed to  $(B_i,F_i)$-$\auti{1}$ which \vplayer{3} must honor. 
If \vplayer{3} moves to $v'$ she picks $(B_j,F_j)$-$\auti{3}$, if she moves to $v'''$ she picks $(B_\ell,F_\ell)$-$\auti{3}$.

The same principle applies for the next rounds.
After \round{0}[k-1], for each automaton a pair has been fixed to which the "Verifier-players" stick for the rest of the play.


We end by substantiating the claim in \cref{fn_nf} on \cpageref{fn_nf} that our construction can also be applied to formulas with arbitrary, i.e., not strictly alternating, trace quantifier prefixes.

\begin{remark}
\label{remark_nonstrictalternation}
    In case of of arbitrary trace quantifier prefixes, we have one player in the coalition for each maximal block of existentially quantified "trace variables" in the prefix, who selects "traces" for these "variables".  
    For example, for a formula of the form
    \[\forall \pi_0 \exists \pi_1 \exists \pi_2\forall \pi_3 \exists \pi_4 \exists \pi_5 .\ \psi\]
    with trace quantifier-free $\psi$, we have just two players (say, Player~$12$ and Player~$45$) in the coalition. Player~$12$ is in charge of the "variables"~$\pi_1$ and $\pi_2$, and Player~$45$ is in charge of the "variables"~$\pi_4$ and $\pi_5$.
    These two players are in a coalition against "Falsifier", who is in charge of the "variables"~$\pi_0$ and $\pi_3$.
    In each round, "Falsifier" picks a move for $\pi_0$, then Player~$12$ picks a move for $\pi_1$ and a move for $\pi_2$,
    then "Falsifier" picks a move for $\pi_3$, and then Player~$45$ picks a move for $\pi_4$ and a move for $\pi_5$.
    
    Note that just adding universally quantified dummy "trace variables" between any pair of consecutive existentially quantified "trace variables" is logically equivalent, but increases the complexity of our approach (as that depends on the number of players).
\end{remark}

\section{Conclusion}

We have presented the first sound and complete game-based algorithm for full \hyltl (even \hyqptl) model-checking via multi-player games with hierarchical information, thereby extending the prophecy-based framework of Coenen et al.\ and Beutner and Finkbeiner from the $\forall^*\exists^*$-fragment to arbitrary quantifier prefixes. 
Similarly, we have extended Winter and Zimmermann's game-based characterization of the existence of computable Skolem functions witnessing $\tsys \models \phi$. 
One way of understanding our result is that we compute finite-state implementations of Skolem functions via transducers with $\omega$-regular lookahead.

In further work, we aim to exhibit lower bounds on the number of prophecies needed for completeness.
Recall that \hyltl model-checking is \tower-complete.
Even an elementary number of prophecies, while unlikely, would not contradict any complexity-lower bounds, as solving multi-player games of hierarchical information is already \tower-complete (in the number of players).
However, as it is, the number of prophecies we use is only bounded by an $n$-fold exponential, where $n$ is linear in the number of quantifier alternations (taking both trace and propositional quantification into account), as the size of the automata~$\auti{i}$ grows like that.

One of the motivations for the introduction of the game-based approach using prophecies is the need for complementation of $\omega$-automata in the classical automata-based \hyltl model-checking algorithm. 
For $\forall^*\exists^*$-formulas, the game-based approach \myquot{only} requires deterministic $\omega$-automata, which can be directly computed from the quantifier-free part, which is essentially an \ltl formula.
On the other hand, the automata-based approach needs automata complementation, already for $\forall^*\exists^*$-formulas.
For arbitrary quantifier alternations, both approaches require complementation.
In future work, we aim to study whether nondeterministic or history-deterministic automata are sufficient to construct complete prophecies.

Finally, we are currently extending the techniques developed here to even more expressive logics, e.g., \hyrechml, recursive Hennessy-Milner logic with trace quantification~\cite{DBLP:conf/forte/AcetoAAF22}. 

\bibliographystyle{plain}
\bibliography{bib}

\end{document}